\documentclass[authoryear,12pt]{article}

\usepackage[english]{babel}
\usepackage{array}
\usepackage{rotating}
\usepackage{color}
\usepackage{setspace}
\usepackage{url}
\usepackage{multirow}
\usepackage{epstopdf}
\usepackage{float}
\usepackage{subcaption}
\RequirePackage{natbib}
\RequirePackage{amsthm,amsmath,amsfonts,amssymb}
\usepackage{mathrsfs}
\RequirePackage[colorlinks,citecolor=blue,urlcolor=blue]{hyperref}
\RequirePackage{graphicx}
\usepackage[bottom]{footmisc}
\usepackage{enumitem}

\newcounter{magicrownumbers}

\RequirePackage{geometry}
\RequirePackage{t1enc}
\RequirePackage{enumerate}
\RequirePackage{epsfig}
\RequirePackage{makeidx}
\RequirePackage{graphicx}
\RequirePackage{listings}
\RequirePackage{gensymb}

\usepackage[utf8]{inputenc}
\usepackage[utf8]{luainputenc}
\usepackage[english]{babel}

\usepackage{xcolor}

\makeindex
\newcommand{\namelistlabel}[1]{\mbox{#1}\hfil}

\newtheorem{remark}{\bf Remark}[section]

\newtheorem{lemma}{\bf Lemma}[section]
\newtheorem{theorem}{\bf Theorem}[section]
\newtheorem{result}{\bf Result}[section]

\newtheorem{illustration}{\bf Illustration}[section]

\theoremstyle{remark}
\newtheorem{example}{\bf Case}

\usepackage{myCmdShortcutPaperI}
\usepackage{myCmdShortcutPaperI_extra}

\begin{document}
	\thispagestyle{empty}

\pagenumbering{arabic}
\newpage

\title{Optimal designs for some bivariate cokriging models}

\author{{\small Subhadra Dasgupta$^1$, Siuli Mukhopadhyay$^{2}$ and Jonathan Keith$^3$}\\%\vspace{0.4cm} 
	{\small \it $^1$ IITB-Monash Research Academy, India}\\%\vspace{0.4cm}
	{\small \it $^2$ Department of Mathematics, Indian Institute of Technology Bombay, India}\\%\vspace{0.4cm}
	{\small \it $^3$ School of Mathematics, Monash University, Australia}}

\date{}

\maketitle

\noindent {\bf Abstract}

This article focuses on the estimation and design aspects of a bivariate collocated cokriging experiment. For a large class of covariance matrices, a linear dependency criterion is identified, which allows the best linear unbiased estimator of the primary variable in a bivariate collocated cokriging setup to reduce to a univariate kriging estimator. Exact optimal designs for efficient prediction for such simple and ordinary reduced cokriging models with one-dimensional inputs are determined. Designs are found by minimizing the maximum and the integrated prediction variance, where the primary variable is an Ornstein-Uhlenbeck process. For simple and ordinary cokriging models with known covariance parameters, the equispaced design is shown to be optimal for both criterion functions. The more realistic scenario of unknown covariance parameters is addressed by assuming prior distributions on the parameter vector, thus adopting a Bayesian approach to the design problem. The equispaced design is proved to be the Bayesian optimal design for both criteria. The work is motivated by designing an optimal water monitoring system for an Indian river.

\vspace{.5em}

	\section{Introduction}

	Kriging is a method for estimating a variable of interest, known as the \textit{primary variable}, at unknown input sites. When multiple responses are collected, multivariate kriging, also known as cokriging, is a related method for estimating the variable of interest at a specific location using measurements of this variable at other input sites along with the measurements of
%we cokrig the variable of interest at a specific location with data about itself and about 
\textit{auxiliary/secondary variables}, which may provide useful information about the primary variable \citep{Myers_1983, Myers_1991,  Book_Wackernagel,  Book_Chiles_Delfiner}. For example, consider a water quality study in which a geologist is interested in estimating pH levels (primary response) at several unsampled locations along a river, but auxiliary information such as phosphate concentration or amount of dissolved oxygen may facilitate more accurate estimates of pH levels. We may also consider a computer experiment,  where the engineering code produces the primary response and its partial derivatives. The partial derivatives (secondary variables) provide valuable information about the response \citep{Book_Santner}. This scenario is typical when the responses measured are correlated, both non-spatially (at the same input sites) and spatially (over different sites, particularly those close to each other).

Very little is known about designs for such cokriging models.  \cite{Cokriging_Li_Zimmerman_2015}, \cite{madani2019comparison}, \cite{Bueso_etal_1999}, \cite{Le_Zidek_1994}, \cite{CASELTON_ZIDEK_1984} developed optimal designs for multivariate kriging models or multivariate spatial processes, however the designs were all based on numerical simulations. 
The key difficulty in using such multivariate models is specifying the cross-covariance between the different random processes. Unlike direct covariance matrices, cross-covariance matrices need not be symmetric; indeed, these matrices must be chosen in such a way that the second-order structure always yields a non-negative definite covariance matrix \citep{Genton_Kleiber_2015, subramanyam2004equivalence}. A broad list of valid covariance structures for multivariate kriging models has been proposed by \cite{Cokriging_Li_Zimmerman_2015}.%In this article, we focus on a wide variety of valid covariance structures used for cokriging models (see Table \ref{TableCovarianceStructures}). Some of these structures have been used before for multivariate kriging models by \cite{Cokriging_Li_Zimmerman_2015}. 
%We derive the validity of the 'Generalized Markov' type covariance structure, which is an improvement of 'Markov type' covariance structure, as outlined in \cite{Almeida_Journel_1994}, \cite{Journel_1999}, \cite[Chapter~5]{Book_Chiles_Delfiner}. \\

In this article, we address two issues for bivariate cokriging experiments, (i) estimation of the primary variable and (ii) determining optimal designs by minimizing the mean squared error of the estimation. In the first couple of sections, we discuss simple and ordinary bivariate collocated cokriging models, the various covariance functions available in the literature for such models, and their estimation aspects. Specifically, we consider two stationary and isotropic random functions, $Z_1$ and $Z_2$ over $\mathcal{D} \subseteq \mathbb{R} $, where $Z_1$ is the primary variable and $Z_2$ is the secondary/auxiliary variable. 
Our main interest is in the prediction of  $Z_1$, at a single location, say $x_0$, in the region of interest. For defining covariance matrices for the bivariate responses, we mainly utilize two families of stationary covariances, namely the generalized Markov-type and the proportional covariance functions. The generalized Markov-type covariance, an extended version of Markov-type covariance, is a new function proposed in this article.  Along with the generalized Markov-type and proportional covariances, the other covariance types mentioned by \cite{Cokriging_Li_Zimmerman_2015} are also studied. We prove a linear dependency condition under which the best linear unbiased predictor (BLUP) of $Z_1(x_0)$ in a bivariate cokriging model is shown to be equivalent to the BLUP in a univariate kriging setup. A wide class of covariance functions is identified which allows this reduction. 

In the later part of the article, we determine optimal designs for some cokriging models, particularly those for which the reduction holds true. We consider the maximum and the integrated cokriging variance of $Z_1(x_0)$ as the two design criterion functions. The primary variable is assumed to have an isotropic exponential covariance, that is, it satisfies $Cov[Z_{1}(x), Z_{1}(x^\prime)] = \sigma_{11} \; {e}^{-\theta |x-x^\prime| }$ with marginal variance $\sigma_{11} > 0 $ and the exponential parameter $\theta> 0 $. Note, $Z_1(x_0)$ is also called an Ornstein–Uhlenbeck process \citep{Kriging_Antognini_Zagoraiou_2010}. For known covariance parameters in simple and ordinary cokriging models, we prove that the equispaced design minimizes the maximum and integrated prediction variance, that is, it is both G-optimal and I-optimal. In real life, however, covariance parameters are most likely unknown. To address the dependency of the design selection criterion on the unknown covariance parameters, we assume prior distributions on the parameter vector and instead determine pseudo-Bayesian optimal designs. The equispaced design is also proved to be the Bayesian I- and G-optimal design.

%{\textcolor{red}{Though many researchers have studied theoretical exact D- and I- optimal designs for univariate kriging experiments with an exponential covariance structure,  no such work is available in the context of bivariate cokriging experiments. Also, this is the first work to determine theoretical Bayesian I- and G- optimal designs for  unknown covariance parameters. }}
The original contributions of this article include (i) a linear dependency condition for reduction of collocated bivariate kriging estimators to a kriging estimator, (ii) the generalized Markov-type covariance, {{(iii) G-optimal designs for known covariance parameters and G-optimal Bayesian designs, for such simple and ordinary reduced bivariate cokriging models and (iv) I-optimal Bayesian designs. }}

 {{We stress that our sole objective is to find theoretical, exact optimal designs, not numerical designs, for bivariate cokriging models. For this reason, we consider only the exponential covariance structure for the primary variable $Z_1$. Note no theoretical exact optimal designs for covariance structures other than the exponential covariance are currently available in the statistical literature. }}

Many researchers have studied  D- and I-optimal designs for univariate kriging experiments with an exponential covariance structure.
%Exact optimal designs for the location scale model were considered by \cite{Boltze_Nather_1982}, \cite{Nather_1985}, \cite[chapter~4]{Nather_1985_Book}, \cite{Muller_Pazman_2001}, \cite{Muller_Pazman_2003} and \cite{kriging_Zimmerman_2006}. 
For single responses and one-dimensional inputs, \cite{kiselak2008equidistant}, \cite{Kriging_Zagoraiou_Antognini_2009},  \cite{Kriging_Antognini_Zagoraiou_2010} proved that equispaced designs are optimal for trend parameter estimation with respect to average prediction error minimization and the D-optimality criterion. For the information gain (entropy criterion) also, the equispaced design was proved to be optimal by \cite{Kriging_Antognini_Zagoraiou_2010}. \cite{kriging_Zimmerman_2006} studied designs for universal kriging models and showed how the optimal design differs depending on whether covariance parameters are known or estimated using numerical simulations on a two-dimensional grid. \cite{diggle} proposed Bayesian geostatistical designs focusing on efficient spatial prediction while allowing the parameters to be unknown. Exact optimal designs for linear and quadratic regression %universal kriging 
models with one-dimensional inputs and error structure of the autoregressive of order one form were determined by \cite{Dette_et_al_2008}. This work was further extended by \cite{Dette_et_al_2013} to a broader class of covariance kernels, where they also showed that the arcsine distribution is universally optimal for the polynomial regression model with correlation structure defined by the logarithmic potential. \cite{baran2013optimal}  and \cite{Baran_Stehlik_2015} investigated optimal designs for parameters of shifted Ornstein-Uhlenbeck sheets for two input variables. 
%\cite{baran2013optimal} showed that for a Gaussian response, optimal designs corresponding to entropy maximization is equispaced, however, the optimal design minimizing the  integrated mean square error (IMSPE) criterion may not be equispaced. For trend parameter estimation, \cite{Baran_Stehlik_2015} showed that the equispaced design is D- optimal. The superiority of K-optimal designs  over D-optimal designs for Ornstein Uhlenbeck was discussed by \cite{baran2017k}. 
More recently, \cite{sikolya2019optimal} worked with the prediction of a complex Ornstein-Uhlenbeck process and derived the optimal design with respect to the entropy maximization criterion. 
%A new approach to design correlated responses by considering continuous time models were proposed recently by \cite{Dette2017new,Dette2016optimal}. 
%\cite{Dette2017new} gave a new approach for constructing an efficient estimator for single parameter and multiparameter; continuous time, regression models with correlated responses. While, \cite{Dette2016optimal} introduced a more generalized; signed least square estimator with respect to which the optimal designs are obtained by minimizing the variance of the estimator for a broad class of covariance kernels.
%For multivariate geostatistical models, optimal designs based on minimization of the mean squared error or the entropy function have been studied. Designs in the presence of unknown covariance parameters were considered by  \cite{Cokriging_Li_Zimmerman_2015}.
%However, as mentioned earlier most of the literature on designs  in a multivariate setting
% \citep{Cokriging_Li_Zimmerman_2015, Bueso_etal_1999, Le_Zidek_1994, CASELTON_ZIDEK_1984} 
%propose numerical optimal designs. To the best of our knowledge, this is the first article which theoretically determines exact optimal designs for bivariate cokriging models. 

In Sections \ref{Model} and \ref{SectionCovarianceFunction} we introduce bivariate cokriging models and the related functions, respectively. The linear dependency condition which allows the BLUP of a cokriging model to reduce to the BLUP of a kriging model is discussed in Section \ref{SectionModelReduction}.  In Section \ref{OptimalDesign_KnownParameter}, we discuss optimal designs for some cokriging models with known and unknown parameters. An illustration using a water quality data set is provided in Section \ref{Illustration_cokriging_known_parameter}. Concluding remarks are given in Section \ref{SectionConcludingRemarks}.
	
\section{Cokriging models and their estimation} \label{Model}
	
	In this section, multivariate kriging models along with their direct covariance and cross-covariance structures are defined. Our focus is on bivariate processes with one-dimensional inputs. Consider two simultaneous random functions $Z_{1}(\cdot)$ and $Z_{2}(\cdot)$, where $Z_{1}(\cdot)$ is the primary response and $Z_{2}(\cdot)$ the secondary response. 

	We assume both  responses are observed over the region $\mathcal{D} \subseteq \mathbb{R} $. In multivariate studies, usually the sets of points at which different random functions are observed might not coincide, but in the case that it does, the design is said to be completely collocated or simply collocated \citep{Cokriging_Li_Zimmerman_2015}. In this article, we work with a completely collocated design and consider that $Z_1(\cdot)$ and $Z_2(\cdot)$ are both sampled at the same set of points $ \mathcal{S} =  \{ x_{1}, x_{2},\ldots, x_{n} \}$, where $\mathcal{S} \subseteq \mathcal{D} \subseteq \mathbb{R} $. 
		We consider $\VectZi$ to be the $n \times 1$ vector of all observations for the random function $Z_{i}(\cdot)$ for $i=1,2$. These random functions are characterized by their mean and covariance structures, with $E[Z_{i}(x)] = m_{i}(x)$ and $Cov(Z_{i}(x), Z_{j}(x')) = \mathcal{C}_{ij}(x,x'), \text{ for } x,x' \in \mathcal{D} \text{ and }i,j = 1,2$. The underlying linear model is given by: 
	\begin{align}
	\begin{pmatrix}
	\VectZOne \\
	\VectZTwo
	\end{pmatrix} &= 
	\begin{pmatrix}
	\MatrixFOne & \pmb{0}\\
	\pmb{0} & \MatrixFTwo
	\end{pmatrix}   
	\begin{pmatrix}
	\VectPiOne\\
	\VectPiTwo
	\end{pmatrix} 
	+
	\begin{pmatrix}
	\VectEpsilonOne \\
	\VectEpsilonTwo
	\end{pmatrix}, \label{LinearModelFrom1}
	\end{align} 
	where $\MatrixFI$ is the $n \times p_{i}$ matrix, with its $k^{th}$ row given by $\Vectfi^{T}(x_{k})$, $\Vectfi(x)$ is the $p_{i} \times  1 $ vector of known basis drift functions $f_{i}^l(.)$ for $l=0,\ldots,p_{i}$ and $\VectPiI$ is the $p_{i} \times 1$ vector of parameters. From equation \eqref{LinearModelFrom1} we see $m_{i}(x) = \Vectfi^{T}(x) \VectPiI$ for  $i=1,2 \text{ and } x \in \mathcal{D}$. We assume $\VectEpsilonI$ to be a zero mean column vector of length $n$ corresponding to the random variation of $\VectZi$. The error covariance is $Cov(\epsilon_{i}(x), \epsilon_{j}(x')) = Cov(Z_{i}(x), Z_{j}(x')) = \mathcal{C}_{ij}(x,x'), \text{ for } x,x' \in \mathcal{D} \text{ and }i,j = 1,2$. 
Using matrix notation, the model in equation \eqref{LinearModelFrom1} can be rewritten as:
	\begin{align}
	\VectZ = \MatrixF \VectPi + \VectEpsilon \label{LinearModelForm2},
	\end{align}
		where $\VectZ = (\VectZOne^T , \VectZTwo^T )^T$ is a $2n \times 1$ vector, $\VectEpsilon = (\VectEpsilonOne^T , \VectEpsilonTwo^T)^T$, $\VectPi = (\VectPiOne^T , \VectPiTwo^T)^T$, and $\MatrixF = 		
	\begin{pmatrix}
	\MatrixFOne & \pmb{0}\\
	\pmb{0} & \MatrixFTwo
	\end{pmatrix}$. 
	%%   & Cov(Z_{i}(x), Z_{j}(x')) = \mathcal{C}_{ij}(x,x'), \text{ for } x,x' \in \mathcal{D} \text{ and }i,j = 1,2. \label{CovarianceFunctionRandomProcess}
	%%   \end{align}		
	% Suppose we are interested in predicting  the value of the primary random function $Z_{1}(\cdot)$ at $x_{0} \in \mathcal{D}$, where the true value of $Z_1(x_0)$ is denoted by $ Z_{0}$, that is $Z_{1}(x_{0}) \equiv Z_{0}$ using a $BLUP$. 
	We are interested in predicting  the value of the primary random function $Z_{1}(\cdot)$ at $x_{0} \in \mathcal{D}$, using the best linear unbiased predictor (BLUP). The true value of $Z_1(x_0)$ is denoted by $ Z_{0}$, that is, $Z_{1}(x_{0}) \equiv Z_{0}$. 
	%We consider the cokriging estimator of $Z_0$, as suggested by \cite[Chapter~5]{Book_Chiles_Delfiner}, an affine function of all available information on $Z_1(\cdot)$ and $Z_2(\cdot)$ at the $n$ sample points, given by: 
	A cokriging estimator of $Z_0$, as given by \citet[Chapter~5]{Book_Chiles_Delfiner}, is an affine function of all available information on $Z_1(\cdot)$ and $Z_2(\cdot)$ at the $n$ sample points, given by
	$\displaystyle{\sum_{i = 1,2} \VectLambdaI^{T}  \VectZi = \sum_{i = 1,2} \sum_{j=1}^{n}  \lambda_{i j } Z_{i}(x_{j}),  
}$ where %	: 
%: 
%	\begin{align*}
%	Z^{\ast \ast}_{ck} & = \sum_{i = 1,2} \sum_{j=1}^{n}  \lambda_{i j } Z_{i}(x_{j})  
%	= \sum_{i = 1,2} \VectLambdaI^{T}  \VectZi. 
%	\end{align*}	
	%A kriging estimator of $Z_0$, as suggested by \cite[Chapter~3]{Book_Chiles_Delfiner}, is an affine function of available information on $Z_1(\cdot)$ only at the $n$ sample points, and is given by: 
	%\begin{align*}
	%Z^{\ast }_{k} & = \sum_{j=1}^{n}  \lambda_{1 j } Z_{1}(x_{j})  
	%= \VectLambdaOne^{T} \VectZOne, 
	%\end{align*}
	$\VectLambdaI =(\lambda_{i1}, \lambda_{i2},\ldots,\lambda_{in})^T$ is an $n \times 1$ vector of weights for $i=1,2$. The cokriging %and kriging 
	estimators can be shown to be the BLUP of $Z_{0}$  \citep[see][for more details]{Hoef_1993}.
	%  The weights are estimated by setting the bias to zero and minimizing the mean squared prediction error ($MSPE$),
	%    \begin{align*}
	%   E(Z^{\ast \ast} - Z_{0})^{2} & = \sum_{i} \sum_{j} \VectLambdaI^{T} \MatrixCIJ \VectLambdaJ 
	%   - 2 \sum_{i} \VectLambdaI^{T}  \sigmaIO 
	%   + \sigma_{00},
	%   \end{align*}
	
	% Some of the notations that we would use throughout the paper are $ \sigmaIO  = Cov(\VectZi, Z_{0}) \text{ for } i=1,2$, $\sigmaNot^{T} = (\sigmaOneO^{T}, \sigmaTwoO^{T})$  and $ \sigma_{00} = Var(Z_{0},Z_{0})$. Also, the covariance of $\VectZ$ is denoted by $ 
	% \covMat =
	% \begin{bmatrix}
	% \MatrixCOneOne  & \MatrixCOneTwo \\ 
	% \MatrixCTwoOne  & \MatrixCTwoTwo \\
	% \end{bmatrix}.  $
	Some notations we use throughout the paper are:
	$ \sigmaIO  = Cov(\VectZi, Z_{0}) $ for $i=1,2$, $\sigmaNot = (\sigmaOneO^{T}, \sigmaTwoO^{T})^{T}$ and $\sigma_{00} = Cov(Z_{0},Z_{0})$. The covariance matrices are denoted $Cov(\VectZi,\VectZj) = \MatrixCIJ$ for $i,j = 1,2$, and the covariance of the entire vector $\VectZ$ is denoted $\covMat = 
	\begin{bmatrix}
	\MatrixCOneOne  & \MatrixCOneTwo \\ 
	\MatrixCTwoOne  & \MatrixCTwoTwo \\
	\end{bmatrix}$. Note, $\covMat$ is a $2n \times 2n$ matrix. 
	
%	In the rest of the article we have denoted the covariance function by $\CovFunctIJ(\cdot)$ and the covariance matrices by $\MatrixCIJ$ for all $i,j= 1,2$. 
	
	%   \begin{align}
	%  % \text{Define, }  
	%  %  Cov(\mathbf{Z_{i}}, \mathbf{Z_{j}}) &= \MatrixCIJ, \text{ for } i,j = 1,2 \\
	%  \;\; \covMat &=
	%  \begin{bmatrix}
	%  \MatrixCOneOne  & \MatrixCOneTwo \\ 
	%  \MatrixCTwoOne  & \MatrixCTwoTwo \\
	%  \end{bmatrix}.  \label{CovarianceMatrixSigma}
	%  \end{align} 
	%   , where  
	%   \begin{align*}
	%   \VectLambdaI =& (\lambda_{i1}, \lambda_{i2},...,\lambda_{iN})^T. 
	%   \end{align*}
	%   
	%   \noindent Then the unbiased affine \textbf{Cokriging estimator} for estimating $Z_{0}$ is taken as (Chiles and Delfiner \cite{Book_Chiles_Delfiner}), 
	%  
	%   The mean squared prediction error ($MSPE$) is given by (\cite{Book_Chiles_Delfiner})
	% 
	%   As stated earlier, the bias of the estimator is set to 0 (ie. the estimator is unbiased). This provides a set of constraints. Then the variance \textbf{MSPE} is minimized subject to these constraints.\\
	%   In this paper we identify optimal designs for {\em simple} and {\em ordinary} co-kriging, which we briefly define and describe in the following section.
	
	%\subsection{Simple Cokriging and Kriging Estimation} \label{Section_SimpleCokrigingModel}
	\subsection{Estimation in simple cokriging models} 
	\label{Section_SimpleCokrigingModel}

	In a simple cokriging %and kriging 
model, the {means} $m_{i}(x)$ are taken to be constant and known. Thus, without loss of generality, we may assume in such cases that the $Z_{i}$'s are zero mean processes for $ i=1,2$ and therefore in this case $\VectPi = (0, 0 )^{T} $. For known covariance parameters \citep[Chapter~5]{Book_Chiles_Delfiner} the cokriging estimator of $Z_{0}$, denoted by $Z^{\ast \ast }_{sck}$, and the cokriging variance, denoted by $ \sigma^{2}_{sck}(x_{0})$, which is also the mean squared prediction error ($MSPE$) at $x_{0}$, are given by: 
	\begin{align} 
	Z^{\ast \ast }_{sck} &= \sigmaNot^{T} \covMat^{-1} \VectZ, \label{SimpleCokrigingEstimator} \\
	\sigma^{2}_{sck}(x_{0}) &= \sigma_{00} - \sigmaNot^{T}  \covMat^{-1} \sigmaNot. \label{SimpleCokrigingMSPE}
	\end{align}
	
	%Where as for known covariance parameters \citep[Chapter~ 3]{Book_Chiles_Delfiner} kriging estimator of $Z_{0}$, denoted by $Z^{\ast }_{sk}$ and kriging variance, denoted by $ \sigma^{2}_{sk}(x_{0})$, which is also the mean squared prediction error (MSPE) at $x_{0}$, is given by: 
	%\begin{align} 
	%Z^{\ast }_{sk} &= \sigmaOneO^{T} \MatrixCOneOne^{-1} \VectZOne \label{SimplekrigingEstimator} \\
	%\sigma^{2}_{sk}(x_{0}) &= \sigma_{00} - \sigmaOneO^{T}  \MatrixCOneOne^{-1} \sigmaOneO. \label{SimplekrigingMSPE}
	%\end{align}
	
	%\subsection{Ordinary Cokriging and Kriging Estimation} \label{SectionLabelOrdinaryCoKriging}
	\subsection{Estimation in ordinary cokriging models} \label{SectionLabelOrdinaryCoKriging}
	\noindent Another popular model known as ordinary cokriging arises when the means are assumed to be constant but unknown, that is, $m_{i}(x) = \mu_{i}, i=1,2$. In this case $\VectPi = (\mu_{1}, \mu_{2})^{T} $ and the basis drift functions are given by $f_{i}^0(x) = 1$ for $i=1,2$. Hence, $
	\MatrixF = \begin{bmatrix}
	\VectOneN & \VectZeroN \\
	\VectZeroN & \VectOneN
	\end{bmatrix}, $ 	where 	$\VectOneN = (1,1,\ldots,1)_{n \times 1}^{T}$, and $\VectZeroN = (0,0,\ldots,0)_{n \times 1}^{T}$. For known covariance parameters \citep[][Chapter~5]{Hoef_1993,Book_Chiles_Delfiner} the ordinary cokriging estimator of $Z_{0}$, denoted by $Z^{\ast \ast }_{ock}$ and the cokriging variance, denoted by $ \sigma^{2}_{ock}(x_{0})$, which is also the mean squared prediction error ($MSPE$) at $x_{0}$, are given by: 
	\begin{align}
	Z^{\ast \ast}_{ock} &= \sigmaNot^{T} \covMat^{-1} \VectZ + (\VectfNot^{T} - \sigmaNot^{T} \covMat^{-1} \MatrixF) (\MatrixF^{T} \covMat^{-1} \MatrixF)^{-1} \MatrixF^{T} \covMat^{-1} \VectZ, \label{OrdinaryCokrigingEstimatorForm1} \\
	\sigma^{2}_{ock}(x_{0})  &= \sigma_{00} - \sigmaNot^{T} \covMat^{-1} \sigmaNot + (\VectfNot  -  \MatrixF^{T} \covMat^{-1}  \sigmaNot)^{T} (\MatrixF^{T} \covMat^{-1} \MatrixF)^{-1} (\VectfNot  -  \MatrixF^{T} \covMat^{-1}  \sigmaNot), \label{OrdinaryCokrigingMSPEForm1}
	\end{align} where $\VectfNot = (1, \; 0 )^{T}$.

%		$\sigmaNot $ ,$\pmb \sigma_0$, $\pmb C_0$ $\pmb {C_0}$

	\section{Bivariate covariance functions} \label{SectionCovarianceFunction}
		In Section~\ref{Model}, we noted the dependency of the cokriging estimators and their variances on the covariance functions. In this article, we consider only isotropic covariance functions, that is, $\CovFunctIJ(x,x') $ is taken as $ \CovFunctIJ(\norm{x-x'})$ for $ x,x' \in \mathcal{D} $, where $\norm{\cdot}$ is some norm function over $\mathcal{D}$. 
		
		We focus on two families of bivariate covariance functions, namely, i) the generalized Markov-type covariance and ii) the proportional covariance (see \cite{journel1999markov}, \citet[Chapter~5]{Book_Chiles_Delfiner}, \citet[Chapter~9]{Book_sudipto_banerjee2014}). Note, that both of these families allow the primary variable to assume any valid covariance. Therefore we can generate a large number of covariance functions from these two families. Also, we will see that the most popularly used covariances belong to either one of these families. Optimal designs based on some of these covariance functions are discussed later.  
	
		The first family of bivariate covariance functions that we discuss is, the newly proposed generalized Markov-type covariance function. This is an extended form of the Markov-type covariance function mentioned in \citet[Chapter~5]{Book_Chiles_Delfiner} and \cite{journel1999markov}. Suppose the two random functions $Z_1(\cdot)$ and $Z_2(\cdot)$ have respective variances $\sigma_{11}$ and $\sigma_{22}$, where $\sigma_{11},\sigma_{22}>0$ and correlation coefficient  $\rho$, $|\rho|<1$. If $( \sigma_{22}-\rho^{2}  \sigma_{11}) > 0$, then the generalized Markov-type function is given as follows: the cross-covariance function $\CovFunctOneTwo(\cdot)$ is considered to be proportional to $\CovFunctOneOne(\cdot)$ that is, $\CovFunctOneTwo(h)=\rho \CovFunctOneOne(h)$, and the direct covariance for the secondary variable is given by $\CovFunctTwoTwo(h)=\rho^2 \CovFunctOneOne(h)+(\sigma_{22}-\rho^2 \sigma_{11})\CovFunctR(h)$ for some valid correlogram $\CovFunctR(.)$ and for $h \in \mathbb{R}$. Thus, the covariance matrix for the bivariate vector $\VectZ$ under the generalized Markov-type structure has the form:
	\begin{align}
		\covMat &=
		\begin{bmatrix} 
			\MatrixMOne & \rho \MatrixMOne\\
			\rho \MatrixMOne &  \rho^{2} \MatrixMOne+ ( \sigma_{22}-\rho^{2}  \sigma_{11}) \MatrixMR
		\end{bmatrix}, \label{GeneralizedMarkovMatrix}  
	\end{align}
	where $(\MatrixMOne)_{ij} = \CovFunctOneOne(|x_{i}-x_{j}|) $ and $(\MatrixMR)_{ij} = \CovFunctR(|x_{i}-x_{j}|)$ for $i,j = 1,\ldots,n$. The validity of the proposed generalized Markov-type covariance function is discussed in details in \ref{GenMarkovValidity}.

	In the case of proportional covariances function, the direct covariance and cross-covariance of the random functions $Z_{1}(\cdot)$ and $Z_{2}(\cdot)$ are proportional to a single underlying covariance function, say $\CovFunctQ(\cdot)$, that is, $\CovFunctIJ(h) = \sigma_{ij} \CovFunctQ(h)$ for $i,j =1,2$ (see \citet[Chapter~5]{Book_Chiles_Delfiner}, \citet[Chapter~9]{Book_sudipto_banerjee2014}). If, %$[\sigma_{ij}]$
	$\begin{bmatrix}
	\sigma_{11} & \sigma_{12} \\
	\sigma_{21} & \sigma_{22} 
	\end{bmatrix}$ is a positive definite matrix, \cite{Book_Chiles_Delfiner} states that $\CovFunctIJ(\cdot)$ is a valid covariance function and hence $ \covMat $ is a valid covariance matrix. Thus, under the proportional covariance model, 
	\begin{align}
		\covMat = 
		\begin{bmatrix}
			\sigma_{11} \MatrixQ & \sigma_{12} \MatrixQ \\
			\sigma_{21} \MatrixQ & \sigma_{22} \MatrixQ
		\end{bmatrix},  \label{isotropic_proportional_structure} \text{ where } (\MatrixQ)_{ij} = \CovFunctQ(|x_{i}-x_{j}|). 
	\end{align}
	
	Some of the covariance functions, popularly used for bivariate cokriging models are Mat(0.5), Mat(1.5), Mat($\infty$), NS1, NS2, NS3 (listed in Table~\ref{TableCovarianceStructures}) (\cite{Cokriging_Li_Zimmerman_2015}). Note that in fact Mat(0.5), Mat(1.5) and Mat($\infty$) belong to the proportional covariance family while covariance function NS1 belongs to the generalized Markov-type covariance family. Details are given in  Table~\ref{TableCovarianceStructures}.

		\begin{table}[H]
	
		\begin{center}
			{\scriptsize
				\begin{tabular}{ l l p{5cm} m{10cm}   }
					%\multicolumn{5}{c}{} \\
					\hline
					\hline
					& &Bivariate covariance function & Specifications    \\
					\hline
										\hline
					A. &   & Generalized Markov-Type & $	\CovFunctOneOne(0) = \sigma_{11} 	$ \\
					& &$|\rho|<1$				   & $	\CovFunctTwoTwo(\norm{\xOne-\xTwo}) = \rho^{2} \CovFunctOneOne(\norm{\xOne-\xTwo}) + ( \sigma_{22}-\rho^{2}  \sigma_{11}) \CovFunctR(\norm{\xOne-\xTwo})$ \\
					&& $( \sigma_{22}-\rho^{2}  \sigma_{11}) > 0$ & $	\CovFunctTwoOne(\norm{\xOne-\xTwo}) = \rho \CovFunctOneOne(\norm{\xOne-\xTwo})	$ \\
					&&	$\sigma_{11}, \sigma_{22}   > 0$    & $	\CovFunctOneTwo(\norm{\xOne-\xTwo}) = \CovFunctTwoOne(\norm{\xOne-\xTwo})	$ \\
					\hline
					    &a.&  NS1 			& $	\CovFunctOneOne(\norm{\xOne-\xTwo}) = \sigma_{11} \lambda^{\norm{\xOne-\xTwo}}	$ \\
					&	&			& $	\CovFunctTwoTwo(\norm{\xOne-\xTwo}) = \sigma_{22} \lambdac^{2}  \lambda^{\norm{\xOne-\xTwo}} + 
					\sigma_{22} (1-\lambdac^{2})  \lambda^{2\norm{\xOne-\xTwo}}$ \\
					&			&			& $	\CovFunctOneTwo(\norm{\xOne-\xTwo}) = (\sigma_{11} \sigma_{22})^{1/2} \lambdac \lambda^{\norm{\xOne-\xTwo}} $ \\
					&		&		& $	\CovFunctOneTwo(\norm{\xOne-\xTwo}) = \CovFunctTwoOne(\norm{\xOne-\xTwo})	$ \\
					& \multicolumn{3}{c}{(Taking, $\CovFunctOneOne(\norm{\xOne-\xTwo}) = \sigma_{11} \lambda^{\norm{\xOne-\xTwo}}$, $\rho = (\sigma_{11}/ \sigma_{22})^{1/2} \lambdac$,  and $\CovFunctR(\norm{\xOne-\xTwo}) =   \lambda^{2\norm{\xOne-\xTwo}}$ in A.)}   \\
					\hline
					\hline
					
					B. &   & Proportional Covariance & $\CovFunctOneOne(\norm{\xOne-\xTwo}) = \sigma_{11} \CovFunctQ(\norm{\xOne-\xTwo})	$ \\
					& &$(\sigma)_{ij}$ is a positive definite matrix & $ \CovFunctTwoTwo(\norm{\xOne-\xTwo}) = \sigma_{22} \CovFunctQ(\norm{\xOne-\xTwo})$ \\
					& &$\CovFunctQ(\cdot)$ is any valid covariance function											  & $ \CovFunctOneTwo(\norm{\xOne-\xTwo}) = \sigma_{12} \CovFunctQ(\norm{\xOne-\xTwo})$ \\
					&			&								  & $ \CovFunctTwoOne(\norm{\xOne-\xTwo}) = \sigma_{21} \CovFunctQ(\norm{\xOne-\xTwo})$ \\
					\hline
					
					  &b.  & Mat(0.5) & $	\CovFunctOneOne(\norm{\xOne-\xTwo}) = \sigma_{11} \lambda^{\norm{\xOne-\xTwo}}	$ \\
					&&& $	\CovFunctTwoTwo(\norm{\xOne-\xTwo}) = \sigma_{22} \lambda^{\norm{\xOne-\xTwo}}$ \\
					&&& $	\CovFunctOneTwo(\norm{\xOne-\xTwo}) = (\sigma_{11} \sigma_{22})^{1/2} \lambdac \lambda^{\norm{\xOne-\xTwo}}	$ \\
					&&& $	\CovFunctOneTwo(\norm{\xOne-\xTwo}) = \CovFunctTwoOne(\norm{\xOne-\xTwo})		$ \\
						& \multicolumn{3}{c}{(Taking, $\CovFunctQ(\norm{\xOne-\xTwo})	= \lambda^{\norm{\xOne-\xTwo}}	
							$ and $\sigma_{12} = (\sigma_{11} \sigma_{22})^{1/2} \lambdac $ in B.)} \\
					\hline
					  &c.& Mat(1.5) & $	\CovFunctOneOne(\norm{\xOne-\xTwo}) = \sigma_{11} [1-\norm{\xOne-\xTwo}log(\lambda)]\lambda^{\norm{\xOne-\xTwo}}	$ \\
					&		& & $	\CovFunctTwoTwo(\norm{\xOne-\xTwo}) = \sigma_{22} [1-\norm{\xOne-\xTwo}log(\lambda)]\lambda^{\norm{\xOne-\xTwo}}$ \\
					&		& & $	\CovFunctOneTwo(\norm{\xOne-\xTwo}) = (\sigma_{11} \sigma_{22})^{1/2} \lambdac [1-\norm{\xOne-\xTwo}log(\lambda)]\lambda^{\norm{\xOne-\xTwo}}$ \\
					&		& & $	\CovFunctOneTwo(\norm{\xOne-\xTwo}) = \CovFunctTwoOne(\norm{\xOne-\xTwo})		$ \\
					&\multicolumn{3}{c}{					(Taking, $\CovFunctQ(\norm{\xOne-\xTwo})	= [1-\norm{\xOne-\xTwo}log(\lambda)]\lambda^{\norm{\xOne-\xTwo}}	
						$ and $\sigma_{12} = (\sigma_{11} \sigma_{22})^{1/2} \lambdac  $ in B.)} \\
					\hline				
					 & d.  & Mat($\infty$) & $\CovFunctOneOne(\norm{\xOne-\xTwo}) = \sigma_{11} \lambda^{\norm{\xOne-\xTwo}^{2}}	$ \\
					&		&	  & $\CovFunctTwoTwo(\norm{\xOne-\xTwo}) = \sigma_{22} \lambda^{\norm{\xOne-\xTwo}^{2}}$ \\
					&		&	  & $\CovFunctOneTwo(\norm{\xOne-\xTwo}) = (\sigma_{11} \sigma_{22})^{1/2} \lambdac \lambda^{\norm{\xOne-\xTwo}^{2}}$ \\
					&		&	  &$ \CovFunctOneTwo(\norm{\xOne-\xTwo}) = \CovFunctTwoOne(\norm{\xOne-\xTwo})	$
					\\
					&\multicolumn{3}{c}{(Taking, $\CovFunctQ(\norm{\xOne-\xTwo})	= \lambda^{\norm{\xOne-\xTwo}^{2}}	
						$ and $\sigma_{12} = (\sigma_{11} \sigma_{22})^{1/2} \lambdac  $ in B.)}
						\\
					\hline\hline
					C.    & &	NS2  			& $	\CovFunctOneOne(\norm{\xOne-\xTwo}) = \sigma_{11} \lambda^{\norm{\xOne-\xTwo}}	$ \\
					&&& $	\CovFunctTwoTwo(\norm{\xOne-\xTwo}) = \sigma_{22} \lambda^{\norm{\xOne-\xTwo}}$ \\
					&&& $	\CovFunctOneTwo(\norm{\xOne-\xTwo}) = (\sigma_{11} \sigma_{22})^{1/2} \lambdac \lambda^{\alpha \norm{\xOne-\xTwo}} $ \\
					&&& $	\CovFunctOneTwo(\norm{\xOne-\xTwo}) = \CovFunctTwoOne(\norm{\xOne-\xTwo})	$ \\
					&&& where $\alpha  = 0.5, 0.75, 0.9 $ according to whether $ \lambdac= 0.2, 0.5, 0.8$  \\
					\hline		
					D.    &&  NS3 			& $	\CovFunctOneOne(\norm{\xOne-\xTwo}) = \sigma_{11} \lambda^{\norm{\xOne-\xTwo}}	$ \\
					&&& $	\CovFunctTwoTwo(\norm{\xOne-\xTwo}) = \sigma_{22} [1 -\norm{\xOne-\xTwo} log(\lambda) + \norm{\xOne-\xTwo}^{2}(log(\lambda))^{2}/3] \lambda^{\norm{\xOne-\xTwo}}$ \\
					&&& $	\CovFunctOneTwo(\norm{\xOne-\xTwo}) = (\sigma_{11} \sigma_{22})^{1/2} \lambdac [1 -\norm{\xOne-\xTwo} log(\lambda) ] \lambda^{\norm{\xOne-\xTwo}} $ \\
					&&& $	\CovFunctOneTwo(\norm{\xOne-\xTwo}) = \CovFunctTwoOne(\norm{\xOne-\xTwo})	$ \\
					\hline
					\hline
				\end{tabular}				
			}
		\end{center}
			\caption{Various bivariate covariance functions. Note, that $0 < \lambda < 1$, $|\lambdac| < 1$ and $\sigma_{11}, \sigma_{22}   > 0$ .	}
		\label{TableCovarianceStructures}
	\end{table}

	\section{Reduction of cokriging estimators to kriging } \label{SectionModelReduction}
	In this section, we discuss conditions under which the cokriging BLUP for the primary variable is reduced to a kriging BLUP. From Sections \ref{Section_SimpleCokrigingModel} and \ref{SectionLabelOrdinaryCoKriging}, it is not apparent that the cokriging and kriging estimators may be similar, particularly given the potentially non-zero correlation suggesting dependency between $Z_{1}(\cdot)$ and $Z_{2}(\cdot)$. However, in Lemma~\ref{ReductionCokrigToKrig}, we show that a linear dependency condition allows this reduction. Some covariance functions for which the reduction does not hold are also discussed. 
	
	We know that kriging is the univariate version of cokriging. Denoting the simple and ordinary kriging estimator of $Z_0$ by $Z^{\ast }_{sk} $ and $Z^{\ast }_{ok} $, respectively, and the respective variances ($MSPE$) at $x_{0}$ by $ \sigma^{2}_{sk}(x_{0})$ and $ \sigma^{2}_{ok}(x_{0})$, from \cite{Book_Chiles_Delfiner} we have, 
	\begin{align} 
	Z^{\ast }_{sk} &= \sigmaOneO^{T} \MatrixCOneOne^{-1} \VectZOne \label{SimplekrigingEstimator}, \\
	\sigma^{2}_{sk}(x_{0}) &= \sigma_{00} - \sigmaOneO^{T}  \MatrixCOneOne^{-1} \sigmaOneO, \label{SimplekrigingMSPE} \\
	Z^{\ast }_{ok} &= \sigmaOneO^{T} \MatrixCOneOne^{-1} \VectZOne + \dfrac{(1-\sigmaOneO^{T} \MatrixCOneOne^{-1} \VectOneN)(\VectOneN^{T} \MatrixCOneOne^{-1} \VectZOne)}{\VectOneN^{T} \MatrixCOneOne^{-1} \VectOneN} \label{OrdinarykrigingEstimator},\\
%	\sigma^{2}_{ok}(x_{0}) &= \begin{pmatrix}
%	1 & \sigmaOneO^{T}
%	\end{pmatrix}
%	\begin{pmatrix}
%	0 & \VectOneN^{T}\\
%	\VectOneN & \MatrixCOneOne 
%	\end{pmatrix}^{-1} \begin{pmatrix}
%	1 \\ \sigmaOneO
%	\end{pmatrix} \nonumber \\
	\sigma^{2}_{ok}(x_{0}) &= \sigma_{00} - \sigmaOneO^{T} \MatrixCOneOne^{-1} \sigmaOneO + \dfrac{(1-\sigmaOneO^{T} \MatrixCOneOne^{-1} \VectOneN)^2}{\VectOneN^{T} \MatrixCOneOne^{-1} \VectOneN}. \label{OrdinarykrigingMSPE}
	\end{align}

	\begin{lemma} \label{ReductionCokrigToKrig}
		For a collocated bivariate cokriging problem with isotropic covariance structures, if the covariance functions $\mathcal{C}_{11}(.)$ and $\mathcal{C}_{12}(.)$ are linearly dependent; 
		$Z^{\ast \ast }_{sck}$ \eqref{SimpleCokrigingEstimator} is equal to $Z^{\ast }_{sk} $ \eqref{SimplekrigingEstimator} and $Z^{\ast \ast}_{ock}$ \eqref{OrdinaryCokrigingEstimatorForm1} is equal to $Z^{\ast }_{ok}$ \eqref{OrdinarykrigingEstimator}. Consequently, $\sigma^{2}_{sck}(x_{0})$ \eqref{SimpleCokrigingMSPE} and $\sigma^{2}_{ock}(x_{0}) $ \eqref{OrdinaryCokrigingMSPEForm1} are equal to $\sigma^{2}_{sk}(x_{0}) $ \eqref{SimplekrigingMSPE} and $\sigma^{2}_{ok}(x_{0}) $ \eqref{OrdinarykrigingMSPE}, respectively.
	\end{lemma}
	\begin{proof}
		Consider $\covMat^{-1}$, which can be written as:
		\begin{align}
		\covMat^{-1} &= 
		\begin{bmatrix}
		\MatrixCOneOne^{-1} + \MatrixCOneOne^{-1} \MatrixCOneTwo (\MatrixCTwoTwo - \MatrixCTwoOne  \MatrixCOneOne^{-1} \MatrixCOneTwo  )^{-1} \MatrixCTwoOne \MatrixCOneOne^{-1}    & - \MatrixCOneOne^{-1} \MatrixCOneTwo (\MatrixCTwoTwo - \MatrixCTwoOne  \MatrixCOneOne^{-1} \MatrixCOneTwo  )^{-1}\\ 
		- (\MatrixCTwoTwo - \MatrixCTwoOne  \MatrixCOneOne^{-1} \MatrixCOneTwo  )^{-1} \MatrixCTwoOne \MatrixCOneOne^{-1}  & (\MatrixCTwoTwo - \MatrixCTwoOne  \MatrixCOneOne^{-1} \MatrixCOneTwo  )^{-1} \\
		\end{bmatrix}. \nonumber
		\end{align}
		From the isotropy assumption we have $\mathcal{C}_{12}(\cdot) = \mathcal{C}_{21}(\cdot)$, and from the assumption of linear dependence of $\mathcal{C}_{12}(\cdot)$ and $\mathcal{C}_{11}(\cdot)$, we have $\mathcal{C}_{12}(\cdot) = c \; \mathcal{C}_{11}(\cdot)$ for some $c \in \mathbb{R}$. Since our designs are collocated, we may write $\MatrixCOneTwo =  \MatrixCTwoOne$ and $ \MatrixCOneTwo= c \; \MatrixCOneOne  $, which implies $\MatrixCOneTwo\; \MatrixCOneOne^{-1}= c \pmb I_{n}   $ %, so $\MatrixCOneTwo ^{-1}=   \; \MatrixCOneOne^{-1}/c$. 
		Also, note that $\sigmaTwoO = c \; \sigmaOneO $. Hence,
		\begin{align}
		\covMat^{-1} & = \begin{bmatrix}
		\MatrixCOneOne^{-1} + c^{2}(\MatrixCTwoTwo - \MatrixCTwoOne  \MatrixCOneOne^{-1} \MatrixCOneTwo  )^{-1}     & - c (\MatrixCTwoTwo - \MatrixCTwoOne  \MatrixCOneOne^{-1} \MatrixCOneTwo  )^{-1}\\ 
		-c (\MatrixCTwoTwo - \MatrixCTwoOne  \MatrixCOneOne^{-1} \MatrixCOneTwo  )^{-1}  & (\MatrixCTwoTwo - \MatrixCTwoOne  \MatrixCOneOne^{-1} \MatrixCOneTwo  )^{-1} \\
		\end{bmatrix} \label{LemmaEq1} 
		\end{align}
		and 
		\begin{align}
		\sigmaNot^{T} & = (\sigmaOneO^{T}, c \sigmaOneO^{T}) \label{LemmaEq2}.
		\end{align}
		
		For simple cokriging models, substituting \eqref{LemmaEq1} and \eqref{LemmaEq2} in \eqref{SimpleCokrigingEstimator} and \eqref{SimpleCokrigingMSPE}, and after some simple matrix calculations, we note that the expressions for the estimator $Z^{\ast \ast }_{sck}$ and variance $\sigma^{2}_{sck}(x_{0})$ are the same as that of a simple kriging estimator $Z^{\ast }_{sk} $ and its variance $\sigma^{2}_{sk}(x_{0}) $, respectively. 

		Following similar steps for the ordinary cokriging model case, we substitute \eqref{LemmaEq1} and \eqref{LemmaEq2} in \eqref{OrdinaryCokrigingEstimatorForm1} and \eqref{OrdinaryCokrigingMSPEForm1}. The ordinary cokriging estimator and variance can similarly be shown to be the same as that of the ordinary kriging estimator and its variance, respectively. 

	\end{proof}	
	
We study the various covariance functions in Table~\ref{TableCovarianceStructures} and identify for which functions the cokriging estimation problem reduces to a kriging problem, that is, the linear dependency condition is fulfilled. For simplicity and uniformity of notations, from this point on we take $\MatrixP$ as an $n \times n $ matrix and $\sigmaPNot$ as an $n \times 1 $ vector corresponding to any  covariance function $\CovFunctP(\cdot)$. Then, $(\MatrixP)_{ij} = \CovFunctP(|x_{i} - x_{j}|)$ and $ (\sigmaPNot)_{i} = \CovFunctP(|x_{i}-x_{0}|)$ for $i,j = 1, \ldots, n$. We consider, $\MatrixCOneOne = \sigma_{11} \MatrixP$ and $\sigmaOneO  = \sigma_{11} \sigmaPNot$. Using these notations, the kriging expressions in equations \eqref{SimpleCokrigingEstimator}, \eqref{SimpleCokrigingMSPE}, \eqref{OrdinaryCokrigingEstimatorForm1}, and \eqref{OrdinaryCokrigingMSPEForm1} become: 
	\begin{align}
	Z^{\ast }_{sk} &= \sigmaPNot^{T} \MatrixP^{-1} \VectZOne, \label{EQ1}\\
	Z^{\ast }_{ok} &= \sigmaPNot^{T} \MatrixP^{-1} \VectZOne + \dfrac{(1-\sigmaPNot^{T} \MatrixP^{-1} \VectOneN)(\VectOneN^{T} \MatrixP^{-1} \VectZOne)}{\VectOneN^{T} \MatrixP^{-1} \VectOneN} \label{EQ2},\\
	MSPE_{sk}(x_{0}) &= \sigma_{11}  
	\Big( 1 - \sigmaPNot^{T} \MatrixP^{-1}\sigmaPNot  \Big), \label{EQ3}\\
	MSPE_{ok}(x_{0})	&= \sigma_{11}  
	\Bigg( 1 - \sigmaPNot^{T} \MatrixP^{-1}\sigmaPNot   
	+ \dfrac{\Big( 1 -  \VectOneN^{T} \MatrixP^{-1} \sigmaPNot \Big)^{2}}{\VectOneN^{T} \MatrixP^{-1} \VectOneN} 
	\Bigg). \label{EQ4}
	\end{align}
Considering the covariance functions from Table~\ref{TableCovarianceStructures} in details we see:

	\begin{example}  \label{ExampleMarkovType}
		Generalized Markov-Type: %\\
		Here we note $\mathcal{C}_{12}(\cdot)$ and $\mathcal{C}_{11}(\cdot)$ are linearly dependent, that is, $\mathcal{C}_{12}(\cdot) = \rho \mathcal{C}_{11}(\cdot)$. From \eqref{GeneralizedMarkovMatrix}, we may write the cross-covariance matrix as, \\$\covMat =
		\begin{bmatrix} 
		\MatrixMOne & \rho \MatrixMOne\\
		\rho \MatrixMOne & 
		\rho^{2} \MatrixMOne+ ( \sigma_{22}-\rho^{2}  \sigma_{11}) \MatrixMR
		\end{bmatrix}  $ and 
		$\sigmaNot = \begin{bmatrix}
		\sigmaOneO\\
		\sigmaTwoO
		\end{bmatrix}
		=
		\begin{bmatrix}
		\sigmaOneO\\
		\rho \; \sigmaOneO 
		\end{bmatrix} $. Considering $\MatrixP$ and $\sigmaPNot$ to be specified by any valid covariance function $\CovFunctP(\cdot)$, the simple and ordinary cokriging estimators and variances are as in equations \eqref{EQ1}, \eqref{EQ2}, \eqref{EQ3} and \eqref{EQ4}. 
		Thus, for the generalized Markov-type covariance given in Table~\ref{TableCovarianceStructures}, the cokriging estimation reduces to kriging estimation. 
	\end{example}
	
	\begin{example} \label{ExampleIsotropic}
		Proportional covariances: %\\
		In this case the underlying covariance function is given by $\CovFunctQ(\cdot)$ in equation \eqref{isotropic_proportional_structure}. Consider $\CovFunctP(\cdot)= \CovFunctQ(\cdot)$, then from equation \eqref{isotropic_proportional_structure} we obtain,
		$\covMat = 
		\begin{bmatrix}
		\sigma_{11} \MatrixP & \sigma_{12} \MatrixP\\
		\sigma_{21} \MatrixP & \sigma_{22} \MatrixP
		\end{bmatrix}$ and 
		$ \sigmaNot 
		=
		\begin{bmatrix}
		\sigma_{11} \sigmaPNot\\
		\sigma_{12} \sigmaPNot
		\end{bmatrix}$. Here, we have $\sigma_{12}$ = $\sigma_{21}$, due to the isotropy of the covariance function. Since $\mathcal{C}_{12}(\cdot) $ and $\mathcal{C}_{11}(\cdot)$ are linearly dependent, the simple and ordinary cokriging estimators and variances are as in equations \eqref{EQ1}, \eqref{EQ2}, \eqref{EQ3} and \eqref{EQ4}. Thus, for isotropic proportional covariances also, the cokriging estimation reduces to kriging estimation. 		
	\end{example}

So, in particular, we can say that the equivalency of the kriging and cokriging estimation also holds good for Mat(0.5), Mat(1.5), and Mat($\infty$) (as they belong to the proportional covariance family) and NS1 (as it belongs to the generalized Markov type covariance family). However, this reduction does not always hold true for a collocated experiment.  
	
	\begin{example} \label{ExampleNS2}
		NS2 covariance function: In this case, we see that the cokriging estimation is not the same as the kriging estimation. \\
		Consider $\CovFunctP(\norm{h})= \lambda^{ \norm{h}}$ and $\CovFunctPTwo(\norm{h})= \lambda^{\alpha \norm{h}}$. From Table~\ref{TableCovarianceStructures}, we get $\mathcal{C}_{11}(\norm{h}) = \sigma_{11} \CovFunctP(\norm{h}) $, $\mathcal{C}_{12}(\norm{h}) = (\sigma_{11}\sigma_{22})^{1/2} \lambdac  \CovFunctPTwo(\norm{h}) $ and $\mathcal{C}_{22}(\norm{h}) = \sigma_{22} \CovFunctP(\norm{h})$. The $n \times n$ matrices $\MatrixP$, $\MatrixPAlpha$ are given as $(\MatrixP)_{ij} = \lambda^{\norm{x_{i}-x_{j}}}$, $(\MatrixPAlpha)_{ij} = \lambda^{\alpha \norm{x_{i}-x_{j}}}$ and the $n \times 1$ vectors $\sigmaPNot$, $\sigmaPAlpha$ are $(\sigmaPNot)_{i} = \lambda^{\norm{x_{i}-x_{0}}}$, $(\sigmaPAlpha)_{i} = \lambda^{\alpha \norm{x_{i}-x_{0}}}$ for all $i,j= 1,\ldots,n$. This gives rise to the bivariate covariance matrix $\covMat = 
		\begin{bmatrix}
		\sigma_{11} \MatrixP & (\sigma_{11}\sigma_{22})^{1/2} \lambdac \MatrixPAlpha \\
		(\sigma_{11}\sigma_{22})^{1/2} \lambdac \MatrixPAlpha & \sigma_{22} \MatrixP
		\end{bmatrix}$ and 
		$ \sigmaNot 
		=
		\begin{bmatrix}
		\sigma_{11} \sigmaPNot\\
		(\sigma_{11}\sigma_{22})^{1/2} \lambdac \sigmaPAlpha
		\end{bmatrix}$. In this case, 
		\begin{align*}
		Z^{\ast }_{sck} &= \sigmaPNot^{T} \MatrixP^{-1} \VectZOne \\
		& + \lambdac^{2} 
		\big[  \NSTwoTOne   -   \sigmaPAlpha  \big]^{T} \NSTwoTTwo \MatrixPAlpha \MatrixP^{-1} \VectZOne \\
		& - \lambdac (\dfrac{\sigma_{11}}{\sigma_{22}})^{1/2} 
		\big[   \NSTwoTOne - \sigmaPAlpha \big]^{T} \NSTwoTTwo \VectZTwo , \\
		%Z^{\ast }_{ok} &= \sigmaPNot^{T} \MatrixP^{-1} \VectZOne + \dfrac{(1-\sigmaPNot^{T} \MatrixP^{-1}, \VectOneN)(\VectOneN^{T} \MatrixP^{-1} \VectZOne)}{\VectOneN^{T} \MatrixP^{-1} \VectOneN} \label{EQ2},\\
		\text{while } MSPE_{sck}(x_{0}) &= \sigma_{11}   \Big( 1 - \sigmaPNot^{T} \MatrixP^{-1}\sigmaPNot  \Big) \\
		&+ \sigma_{11}  \lambdac^{2}  \Big[  -(\NSTwoTOne)^{T} \NSTwoTTwo \NSTwoTOne  \\
		&+ 2 (\NSTwoTOne)^{T} \NSTwoTTwo \sigmaPAlpha - \sigmaPAlpha^{T} \NSTwoTTwo \sigmaPAlpha \Big]. 
		% MSPE_{ok}(x_{0})	&= \sigma_{11}   \Bigg( 1 - \sigmaPNot^{T} \MatrixP^{-1}\sigmaPNot   
		%+ \dfrac{\Big( 1 -  \VectOneN^{T} \MatrixP^{-1} \sigmaPNot \Big)^{2}}{\VectOneN^{T} \MatrixP^{-1} \VectOneN} 
		%\Bigg).
		\end{align*}
	\end{example}
Similarly, in the case of an NS3 covariance function, it can be shown that the cokriging estimation differs from the kriging estimation.%	Similarly, in case of the covariance structure NS3, it can be shown that the cokriging estimation is different from the kriging estimation. %This ends the first part of this paper. In the following sections we will derive optimal designs for kriging experiments, particularly for the covariance structures where the reduction of cokriging to kriging holds true.
	
	\section{Optimal designs} \label{OptimalDesign_KnownParameter}
	In this section and the following ones, we prove various results for optimally designing collocated bivariate cokriging experiments. The set on which the random functions $Z_{1}(\cdot)$ and $Z_{2}(\cdot)$ are observed is a connected subset of $\mathbb{R}$, denoted by $\mathcal{D}$, while the set on which they are sampled is denoted by $\mathcal{S} =  \{ x_{1}, \ldots, x_{n} \}$, where $\mathcal{S} \subseteq \mathcal{D}$.
	
	%Under certain conditions on covariance structure these designs could be used to design cokriging experiments as well. All the results point that an equispaced design is optimal. 
In the context of finding a design, we are essentially interested in choosing a set of distinct points $\{ x_1,\ldots,x_n \}$ which maximizes the prediction accuracy of the primary response $Z_1(\cdot)$. To choose such a design, the supremum of $MSPE$, denoted as $SMSPE$, where
	\begin{align}
	SMSPE &= \sup_{x_{0} \in D} MSPE(x_{0}),  \label{SMSPE}
	\end{align} 
	or alternatively, an integrated version of $MSPE$ denoted by $IMSPE$, where 
	\begin{align}
	{IMSPE } &= \int\limits_{x_0\in D} MSPE(x_{0}) d(x_{0}), \label{IMSPE}
	\end{align} are minimized.

	Since replications are not allowed, the points are assumed to be ordered, that is, $x_{i} < x_{j}$  for $i < j$, and the distance between two consecutive points is denoted by $d_{i} = x_{i+1} - x_{i}$ for $ i= 1,\ldots, n-1$. For kriging models, since extrapolation should be treated with caution \citep{sikolya2019optimal}, we take an approach similar to \cite{sikolya2019optimal} and \cite{Kriging_Antognini_Zagoraiou_2010}. The starting and end points of the design, $x_{1} $ and $x_{n}$ are considered to be known and given by the extreme ends of the area under observation. This approach in fact reduces the number of variables in the design problem and makes it more simplified. Hence, $\mathcal{D}= [x_{1}, x_{n}] $ and $\sum_{i=1}^{n-1} d_{i} = x_{n} - x_{1}$. We equivalently denote the design by the vector $\xibold= (x_{1},d_{1},d_{2},\ldots,d_{n-1},x_{n})$ in terms of the starting point, consecutive distances between the points, and the end point. 
	
	In this article, for the purpose of finding optimal designs we consider simple and ordinary bivariate collocated cokriging models, with isotropic random functions. The covariance functions belongs to generalized Markov-type or proportional covariance family. For these families of covariance functions, we have seen in the earlier sections that the cokriging to kriging reduction holds true. We also consider that the primary variable $Z_{1}(\cdot)$ is an Ornstein–Uhlenbeck process with exponential parameter $\theta > 0 $ and variance $\sigma_{11}>0$. Hence, $\CovFunctP(|h|)=  {e}^{-\theta |h| }$ would mean $\mathcal{C}_{11}(|h|) = \sigma_{11} \CovFunctP(|h|) $ and the matrix $\MatrixP$ and vector $\sigmaPNot$ are given by $(\MatrixP)_{ij} = {e}^{-\theta|x_{i}-x_{j}|}$ and $(\sigmaPNot)_{i} = {e}^{-\theta|x_{i}-x_{0}|}$ for all $i,j= 1,\ldots,n$ and $x_{0} \in \mathcal{D}$.
		
Note, the optimal designs found in this paper are applicable in particular, to collocated cokriging experiments with Mat($0.5$) or NS1 covariance function as well (as they belong to proportional type and generalized Markov-type family, respectively and for both of these functions, the primary variable has an exponential covariance with exponential parameter $\theta = - log(\lambda)$ as per Table~\ref{TableCovarianceStructures}).

	 %(Note in Table \ref{TableCovarianceStructures} the primary variable for Mat($0.5$), Mat($1.5$) and NS1 has a exponential covariance structure.)
	\subsection{Optimal design results} \label{Paper_I_Section_Optimal design}
	We will show that optimal designs obtained for either criterion (\textit{SMSPE/IMSPE}), for both known and unknown covariance parameters, are equispaced. The following lemma gives the mathematical forms of $MSPE_{sck}(.)$ and $MSPE_{ock}(.)$, and are used in all the results in this article. 	
	\begin{lemma} \label{LemmaSimpleAndOrdinary}
		Consider simple and ordinary bivariate collocated cokriging models, with isotropic random functions. The bivariate covariance functions could be generalized Markov-type or proportional type with the primary variable $Z_{1}(\cdot)$ having an exponential structure, such that $\CovFunctOneOne(h) = \sigma_{11} \; e^{-\theta|h|}$ for  $\sigma_{11}, \theta > 0$. Then, the expressions for MSPE at point $x_{0} \in [x_{i}, x_{i+1}]$ for some $i = 1,\ldots,n-1$ are: 
		\begin{align*}
		MSPE_{sck}(x_{0}) 
		&= \sigma_{11} \dfrac{\left( 1   - e^{-2 \theta a}\right) \left( 1   - e^{-2 \theta (d_{i} - a)}\right)}{\left( 1   - e^{-2 \theta d_{i}}\right)}
		\end{align*} 
		and 
		\begin{align*}
		MSPE_{ock}(x_{0}) 
		&= \sigma_{11}\Bigg[ \dfrac{\left( 1   - e^{-2 \theta a}\right) \left( 1   - e^{-2 \theta (d_{i} - a)}\right)}{\left( 1   - e^{-2 \theta d_{i}}\right)} + \dfrac{1}{\qZero} \Big(1 -  \dfrac{e^{- \theta a } + e^{- \theta( d_{i} - a )  } }{ 1 + e^{- \theta d_{i}  }} \Big)^2 \Bigg], 	
		\end{align*}
		where $ a = x_{0} - x_{i}$ and $\qZero  = \VectOneN^{T} \MatrixP^{-1} \VectOneN.$
	\end{lemma}
	 
	\begin{proof}
		Note that from Lemma~\ref{ReductionCokrigToKrig}, for the above two families of covariance function (the generalized Markov-type covariance and the proportional covariance) the cokriging estimation reduces to a kriging estimation. Using equation \eqref{decomposition3} from \ref{AppendixC}, in equation \eqref{EQ3} and doing simple algebraic computations gives the above expression of $MSPE_{sck}(x_{0})$ (same as $MSPE_{sk}(x_{0})$ in this case). Similarly, using equations \eqref{decomposition3} and \eqref{decomposition4} from \ref{AppendixC}, in equation \eqref{EQ4} gives the above expression of $MSPE_{ock}(x_{0})$ (same as $MSPE_{ok}(x_{0})$ in this case).   
	\end{proof}
	
	\noindent	Note: The $MSPE$ expressions are the same as in Lemma \ref{LemmaSimpleAndOrdinary} when the covariance functions are Mat($0.5$) or NS1 (in that case $\theta = - log(\lambda)$). \\
	
	To reduce the computational complexity, we further claim that a random process over $[x_{1}, x_{n}]$ could be viewed as a process over $[0, 1]$. 
	
	\begin{remark}
		Consider the reduced bivariate collocated cokriging models as in Lemma~\ref{LemmaSimpleAndOrdinary}, defined over $[x_{1}, x_{n}]$ and recorded at $\{x_{1}, \ldots, x_{n}\}$. From the expressions of $MSPE_{sck}$ and $MSPE_{ock}$, we can say that $Z_{1}(\cdot)$ is equivalent to an isotropic process with exponential parameter $(x_{n}- x_{1})\theta$ over $[0, 1]$ and recorded at $\{(x_{i} - x_{1})/(x_{n} - x_{1}), i = 1, \ldots, n\}$. 
	\end{remark}
	\begin{proof}
		We have the design vector $\pmb \xi$ = $(x_{1}, d_{1}, \ldots, d_{n-1}, x_{n})$, where $d_{i} = x_{i+1}- x_{i}$ for $i = 1, \ldots, n-1.$ Then, for $x_{0} \in [x_{i},  x_{i+1}]$ for some $i= 1, \ldots, n-1$, and using Lemma \ref{LemmaSimpleAndOrdinary},
		\begin{align*}
		MSPE_{sk}(x_{0};\pmb  \xi, \theta, \sigma_{11}) 
		&= \sigma_{11} \dfrac{\left( 1   - e^{-2 \theta a}\right) \left( 1   - e^{-2 \theta (d_{i} - a)}\right)}{\left( 1   - e^{-2 \theta d_{i}}\right)}.
		\end{align*} 
		Define a mapping $\chi(\cdot)$ over $[x_1, x_{n}]$ to $[0,1]$, such that, for any point $x \in [x_{1}, x_{n}]$, $\chi(x) = (x-x_{1})/(x_{n}-x_{1})$. Let, $y_{i} = \chi(x_{i}) $ for $ i = 1, \ldots, n$. If we take $g_{i} = d_{i}/(x_{n}-x_{1})$, then the design $\xibold^\ast = (0,g_{1},\ldots, g_{n-1},1)$ specifies the design or the set of points $\{y_{i}: i = 1, \ldots, n\}$, where $y_{1} = 0 $ and $y_{n}=1$. Consider the point $\chi(x_{0}) = (x_{0}-x_{1})/(x_{n}-x_{1}) \in [y_{i}, y_{i+1}]$, then we have
		\begin{align}
		 MSPE_{sk}(x_{0};\pmb \xi, \theta, \sigma_{11}) 
		&= \sigma_{11} \dfrac{\left( 1   - e^{-2 (x_{n}- x_{1})\theta a/(x_{n}- x_{1})}\right) \left( 1   - e^{-2 (x_{n}- x_{1})\theta (d_{i} - a)/(x_{n}- x_{1})}\right)}{\left( 1   - e^{-2 (x_{n}- x_{1})\theta d_{i}/(x_{n}- x_{1})}\right)} \nonumber \\
		& 		= \sigma_{11} \dfrac{\left( 1   - e^{-2 \phi b}\right) \left( 1   - e^{-2 \phi (g_{i} - b)}\right)}{\left( 1   - e^{-2 \phi g_{i}}\right)} = MSPE_{sk}(\chi(x_{0}); \pmb \xi^\ast, \phi, \sigma_{11}), \label{remarkEq2}
		\end{align} 
		where $\phi = (x_{n}- x_{1})\theta$ and $b = \chi(x_{0}) - y_{i} = (x_{0}-x_{i})/(x_{n}-x_{1}) = a/(x_{n}-x_{1})$. From equation \eqref{remarkEq2} and as $\chi(\cdot)$ is a bijective function, we can assert our claim. 
		
		Similar proof can be given for ordinary cokriging.
	\end{proof}
	Hence, if we need to find an $n$ point optimal design with fixed end points for an exponential process with parameter $\theta$ defined over $[x_{1}, x_{n}]$, we can equivalently find the $n$ point optimal design with fixed end points for the exponential process with parameter $\theta(x_{n} - x_{1})$ and defined over $[0,1]$. 
	
	Conversely, if an (optimal) design over $[0,1]$ is given by $\{y_{1}, \ldots, y_{n} \}$, where $y_{1}=0 $ and $y_{n} =1$, we can get the equivalent (optimal) design over $[x_{1}, x_{n}]$ by taking the transformation $x_{i} = (x_{n}-x_{1}) y_{i} + x_{1}$ for $i=1,\ldots,n$. 
	
	So, from now onwards since $\mathcal{D} \subseteq \mathbb{R}$ is connected, without loss of generality we assume $\mathcal{D} = [0,1]$ with $x_{1}=0$ and $x_{n}=1$, which gives $\displaystyle{\sum_{i=1}^{n-1} d_{i} = 1}$ and the design denoted by $\xibold = (d_{1},d_{2},\ldots	d_{n-1})$.

	\subsection{Optimal designs for reduced bivariate simple cokriging model with known parameters} \label{OptimalDesign_SimpleKrigingKnownParameter}
	In this section, we determine optimal designs for a simple cokriging model in Theorems~\ref{Theorem_SK_SMSPE} and \ref{Theorem_SK_IMSPE}.

	\begin{theorem} \label{Theorem_SK_SMSPE} 
		Consider the reduced bivariate simple cokriging models as in Lemma~\ref{LemmaSimpleAndOrdinary}, with the covariance parameters of the primary response, $\theta$ and $\sigma_{11}$, being known. An \textbf{equispaced} design minimizes the $SMSPE_{sck}$. Thus, the equispaced design is the G-optimal design.
	\end{theorem}
	\begin{proof}
		Consider a point $x_{0} \in \mathcal{D}$, such that $x_0 \in [x_i, x_{i+1}]$ for some $i = 1,\ldots,n-1$, then from Lemma~\ref{LemmaSimpleAndOrdinary},    
		\begin{align}
		MSPE_{sck}(x_{0}) 
		& = \sigma_{11} \dfrac{\left( 1   - e^{-2 \theta a}\right) \left( 1   - e^{-2 \theta (d_{i} - a)}\right)}{\left( 1   - e^{-2 \theta d_{i}}\right)}. \nonumber % , \;\; \text{ where } a = x_{0} - x_{i} . \nonumber %\label{Thm1_eq1}	
		\end{align}  		
%		Since, $x_0 \in [x_i, x_{i+1}]$, therefore $a \in [0,d_{i}]$ for $i=1,\ldots, n-1$. Now, consider the function:  		
%		\begin{flalign}
%		W_{i}:& \;\; [0,d_{i}] \; \to \mathbb{R} \;\;\;\;\;\;\;\;\;\;\;\;\;  \nonumber \\%\mbox{ such that }	 \\
%		a & \mapsto	  \dfrac{\left( 1   - e^{-2 \theta a}\right) \left( 1   - e^{-2 \theta (d_{i} - a)}\right)}{\left( 1   - e^{-2 \theta d_{i}}\right)}.  \nonumber	 
%		\end{flalign}
%		We have, 
%		\begin{align}
%		& \dfrac{d W_{i}(a)}{ d a }  = \dfrac{2\theta\left( e^{-2 \theta a}   - e^{-2 \theta (d_{i} - a)}\right)}{\left( 1   - e^{-2 \theta d_{i}}\right)},   \nonumber	   
%		\end{align}
%		where, 
%		\begin{align}
%		& \dfrac{d W_{i}(a)}{ d a }\Bigg|_{a = d_{i}/2} = 0, \label{Thm1_eq2} 
%		\end{align}  
%		and 
%		\begin{align}
%		\dfrac{d^{2} W_{i}(a)}{ d a^{2} }   &=  \dfrac{ -4 \theta^2	 \left( e^{-2 \theta a}   + e^{-2 \theta (d_{i} - a)}\right)}{\left( 1   - e^{-2 \theta d_{i}}\right)} < \; 0.   \label{Thm1_eq3}
%		\end{align} 
		From \ref{AppendixD}, we see that for $x_0 \in [x_i, x_{i+1}]$, $MSPE_{sck}(x_{0})$ is maximized at $x_{0} = x_{i} + \dfrac{d_{i}}{2}$, which is the mid-point of the interval $[x_i, x_{i+1}]$. From equation \eqref{Thm1_eq4} we have,
		\begin{align}
		\sup_{x_0 \in [x_{i},x_{i+1}]} MSPE_{sck}(x_{0}) = & \; \sigma_{11} \frac{1 - e^{-\theta d_i}}{1 + e^{-\theta d_i}}. \nonumber	
		\end{align}
		Consider, $W_{{sup}}(\cdot)$ to be a function defined on $[0,1]$, such that $W_{{sup}}(d) = \dfrac{1 - e^{-\theta d}}{1 + e^{-\theta d}}$. Then $W_{{sup}}(d)$ is an increasing function in $d$, as $W^{\prime}_{{sup}}(d) = \dfrac{2 \theta  e^{- \theta d}}{1 + e^{- \theta d}} > 0 .$ Hence,
		\begin{align}
		SMPSE_{sck} &= \sup_{x_0 \in [0,1]} MSPE_{sck}(x_{0}) \nonumber \\
		&= 		  \max_{i=1, \ldots, n-1} \; \sup_{x_0 \in [x_{i},x_{i+1}]} MSPE_{sck}(x_{0}) \nonumber \\ 
		&=\sigma_{11} \; \max_{i=1, \ldots, n-1} W_{{sup}}(d_{i}) \;\;\;\;\;\;\;\; (\text{from equation } \eqref{Thm1_eq4}) \nonumber \\ 
		%&= \sigma_{11} \; \frac{1 - e^{-\theta \max_i d_i}}{1 + e^{-\theta \max_i d_i}}  \nonumber	\\
		&= \sigma_{11} \; W_{{sup}}(\max_i \; d_i) \label{Thm1_SMSPE_eq} 
		\end{align}
		From equation \eqref{Thm1_SMSPE_eq}, for known $\theta$ and $\sigma_{11}$, the $SMSPE_{sck}$ is a function of $\smash{\displaystyle \max_{i} } \; d_i$. Since $W_{{sup}}(d)$ is an increasing function, therefore $SMSPE_{sck}$ is minimized when $\smash{\displaystyle \max_{i} }\; d_i$ is minimized, which occurs for an equispaced partition.  
	\end{proof}

	\begin{theorem} \label{Theorem_SK_IMSPE}
			Consider the reduced bivariate simple cokriging models as in Lemma~\ref{LemmaSimpleAndOrdinary}, with known covariance parameters $\theta$ and $\sigma_{11}$. An \textbf{equispaced} design minimizes the $IMSPE_{sck}$. Thus, the equispaced design is the I-optimal design.
		\end{theorem}
	\begin{proof}
%		Consider a point $x_{0} \in \mathcal{D}$, from Lemma~\ref{LemmaSimpleAndOrdinary}, if $x_0 \in [x_i, x_{i+1}], $ for some $i = 1,...,n-1$ then,   
From Lemma~\ref{LemmaSimpleAndOrdinary}	we can write,  	
%\begin{align*}
%		MSPE_{sk}(x_{0}) 
%		& = \sigma_{11} \dfrac{\left( 1   - e^{-2 \theta a}\right) \left( 1   - e^{-2 \theta (d_{i} - a)}\right)}{\left( 1   - e^{-2 \theta d_{i}}\right)}, \;\; \text{ where } a = x_{0} - x_{i} .	
%		\end{align*}  
		\begin{align}
		IMPSE_{sck} 		& = \sigma_{11}  \int\limits_{0}^{1}   MSPE_{sck}(x_{0}) d(x_{0}) \nonumber \\
		& = \sigma_{11} \sum_{i=1}^{n-1} \int\limits_{x_{i}}^{x_{i+1}}   MSPE_{sck}(x_{0}) d(x_{0}) \nonumber \\
		%&= \int\limits_{[0,1]}   MSPE(x_{0}) d(x_{0}) \nonumber \\
		&= 		 \sigma_{11} \sum_{i=1}^{n-1} \int\limits_{x_{i}}^{x_{i+1}}   \dfrac{\left( 1   - e^{-2 \theta a}\right) \left( 1   - e^{-2 \theta (d_{i} - a)}\right)}{\left( 1   - e^{-2 \theta d_{i}}\right)}d(x_{0})  \nonumber \\ %\;\; \text{, where } a = x_{0} - x_{i}   \nonumber \\ 
 	%	&= 		 \sigma_{11} \sum_{i=1}^{n-1} \int\limits_{0}^{d{i}}   \dfrac{\left( 1   - e^{-2 \theta a}\right) \left( 1   - e^{-2 \theta (d_{i} - a)}\right)}{\left( 1   - e^{-2 \theta d_{i}}\right)}d(a)   \nonumber \\ 
			&= \sigma_{11} \sum_{i=1}^{n-1}  \Big[ \dfrac{1 + e^{-2 \theta d_{i}}}{1 - e^{-2 \theta d_{i}}}d_{i} -\dfrac{1}{\theta} \Big]   \nonumber \\
		&= \sigma_{11}  \Big[  1 - \dfrac{n-1}{\theta} + 2 \;\Phi(\xibold)  \Big], %\;\;\;\;\;\;\;\;\;\;\;\;\;\;\; (\text{See \eqref{IMSPE_OCK_Main_termI} in Appendix \ref{AppendixB}})
		\end{align}
		where $\displaystyle{\Phi(\xibold) = \sum_{i=1}^{n-1} \phi(d_{i})}$ and $\phi(d) = \dfrac{d}{e^{2\theta d } - 1 }.$			\\
		
Using \ref{AppendixE}, we can say that $IMSPE_{sck}$ is a Schur-convex function and hence it is minimized for an equispaced design, that is, $d_i=\frac{1}{n-1}$ for all $i=1,\ldots, n-1$.%		To show that the equispaced design minimizes the IMSPE, we prove that the IMSPE is a Schur-convex function. First note, IMSPE is a symmetric function, that is, it is  permutation invariant in the $d_{i}$'s. Next we show $\dfrac{\partial IMSPE}{\partial d_{i}} $ is an increasing function in $d_{i}$ for $i= 1,\ldots,n$. We have, 
	\end{proof}
	%%%%%%%%%%%%%%%%%%%%%%%%%%%%%%%%%%%%%%%
	%\subsubsection{Optimal designs based on SMSPE}
	
	\subsection{Optimal designs for reduced bivariate simple cokriging models with unknown parameters} \label{OptimalDesign_unKnownParameter}
	%	\subfile{Sections/5_Parameter_Interval_Estimate}
	
	In real life, while designing an experiment, the exponential covariance parameters $\theta$ and $\sigma_{11}$ are usually unknown with very little prior information. In this section, we discuss optimal designs for simple cokriging models with the primary response having an exponential covariance structure but with unknown parameters. To address the dependency of the design selection criterion on the unknown covariance parameters, we assume prior distributions on the parameter vector and instead propose pseudo-Bayesian optimal designs. %; a related approach is discussed in \cite{stehlik2015robust}, which gives a compound criterion for finding optimal designs for trend parameter and integrated prediction variance accounting for the unknown covariance parameter. 
	The prior distributions on the covariance parameters are incorporated into the optimization criteria by integrating over these distributions. This approach is known as the pseudo-Bayesian approach to optimal designs and has been used previously by \cite{Chaloner_Larntz_1989}, \cite{Dette_Sperlich_1996}, \cite{Woods_Ven_2011}, \cite{Mylona_Goos_Jones_2014}, \cite{Singh_Mukhopadhyay_2016} and \cite{Singh_Mukhopadhyay_2019}. The Bayesian approach has been seen to yield more robust optimal designs which are less sensitive to fluctuations of the unknown parameters than locally optimal designs.
	
	We start by assuming  $\theta$ and $\sigma_{11}$ are independent and their respective distributions are $r(\cdot)$ and $t(\cdot)$. 
	%			We have seen (in section 4) that for bivariate co-kriging problems with generalized markov type or proportional covariance structure, when the covariance parameters are known, an equispaced design is optimal. The parameters $\theta$ and $\sigma_{11}$ are usually unknown. However,	in a real scenario we replace the parameters by their estimates in the expression for \textbf{IMSPE}. In section 5 we use maximum likelihood estimation to obtain the parameter estimates. It would be important to understand the nature of optimal design and the design criterion when the parameters are unknown. 
	A very high value of $\theta$ would mean that the covariance matrix for $Z_{1}(\cdot)$ is approximately an identity matrix, implying zero dependence among neighboring points. Since this is not reasonable for such correlated data, we assume $0<\theta_{1}<\theta<\theta_{2}<\infty$. 
	
	Using a pseudo-Bayesian approach as in \cite{Chaloner_Larntz_1989} we define risk functions corresponding to each design criterion as:
	\begin{align}
			\mathcal{R}_{1}(\xibold) & = E[SMSPE_{sck}(\theta, \sigma_{11}, \xibold)] \text{ and } \label{risk_II} \\
		\mathcal{R}_{2}(\xibold) & = E[IMSPE_{sck}(\theta, \sigma_{11}, \xibold)].\label{risk_I}
	\end{align}
	
	Our objective is to select the designs that minimize these risks.
	
%			We take a look at \cite{Kriging_Zagoraiou_Antognini_2009, Kriging_Antognini_Zagoraiou_2010, Baran_Stehlik_2015}, which finds optimal designs using Fisher information matrix for constant unknown trend parameter, with the random function having an exponential covariance with known parameter. It is shown, that for the same problem the strategy optimal design for covariance parameter is either conflicting with that of trend parameter estmation or there is no admissible design. From these finding it seems reasonable to guess that the optimal Bayesian designs might be different from equispaced. However, in this section and also in Section~\ref{SectionLabelOrdinaryCoKriging} we will see the optimal design for known parmeter and unknown paramter case is the same and, hence we could say the designs found in this paper are robust when prediction accuracy of the primary random process is the main objective. 
		
	\begin{theorem} \label{Theorem_SK_SMSPE_UnknownParameter} 
	Consider the reduced bivariate simple cokriging models as in Lemma~\ref{Theorem_SK_SMSPE}. The parameters $\theta$ and $\sigma_{11}$ are assumed to be unknown and independent with prior probability density functions $r(\cdot)$ and $t(\cdot)$, respectively. The support of $r(\cdot)$ is of the form $(\theta_{1},\theta_{2})$, where $0 < \theta_{1} < \theta_{2} < \infty $. Then, an equispaced design is optimal with respect to the risk function $\mathcal{R}_{1}(\xibold)$.
	\end{theorem}
	\begin{proof}				
		From \eqref{Thm1_SMSPE_eq} we can write,
		\begin{align}
		SMPSE_{sck} &= \sigma_{11} \; W_{{sup}}(\max_i \; d_i).  \nonumber 
		\end{align}
		Thus,
		\begin{align}
		\mathcal{R}_{1}(\xibold) &= \int\limits_{0}^{\infty} \int\limits_{\theta_{1}}^{\theta_{2}} 
		\sigma_{11}  \; W_{{sup}}(\theta, \max_i d_i)  \; \;
		r(\theta) \; t(\sigma_{11}) \;\;
		d(\sigma_{11})\;d(\theta)	\nonumber \\
		&= \int\limits_{0}^{\infty} \sigma_{11}  \; t(\sigma_{11}) \; d(\sigma_{11})
		\int\limits_{\theta_{1}}^{\theta_{2}} \; W_{{sup}}(\theta, \max_i d_i)  \;
		r(\theta) 
		\;d(\theta)	\nonumber \\
		&= 	E_{t}[\sigma_{11}] \; 	\int\limits_{\theta_{1}}^{\theta_{2}} \; W_{{sup}}(\theta, \max_i d_i) \;
		r(\theta) 
		\;d(\theta).		\label{SMSPEPseudoBayesRiskMinimiztionEquation} 
		%&= 	E_{g}[\sigma_{11}] \; 	 \int_{m_{1}}^{m_{2}} \;  W_{{sup}}(\theta, \max_i d_i)
		%\mathbf{f}(\theta) 
		%\;d(\theta).	\nonumber \\
		%&= 	E_{g}[\sigma_{11}] \; 	 \int_{m_{1}}^{m_{2}} \;  W_{{sup}}(\theta, \max_i d_i)
		%\mathbf{f}(\theta)  
		%\;d(\theta).	
		\end{align}
		As $W_{{sup}}(\theta, d)$ is an increasing function of  $d$, equation \eqref{SMSPEPseudoBayesRiskMinimiztionEquation} shows $\mathcal{R}_{1}$ is minimized for an equispaced design, since $\smash{\displaystyle \max_{i} }\; d_i$ is minimized for an equispaced design.
	\end{proof}
	
	\begin{theorem} \label{Theorem_SK_IMSPE_UnknownParameter}
			Consider the reduced bivariate simple cokriging models as in Lemma~\ref{Theorem_SK_SMSPE}. The parameters $\theta$ and $\sigma_{11}$ are assumed to be unknown and independent with prior probability density functions $r(\cdot)$ and $t(\cdot)$, respectively. The support of $r(\cdot)$ is of the form $(\theta_{1},\theta_{2})$, where $0 < \theta_{1} < \theta_{2} < \infty $. Then, an equispaced design is optimal with respect to the risk function $\mathcal{R}_{2}(\xibold)$.
	\end{theorem}
	\begin{proof}				
		Consider $\mathcal{R}_{2}: \mathcal{I}^{n-1} \longrightarrow \mathbb{R}$, where $\mathcal{I} = [0,1]$. 
		%We consider $\mathbf{f}(\cdot)$ and $\mathbf{g}(\cdot)$ to have finite support. The risk
		$\mathcal{R}_{2}(\cdot)$ is symmetric on $\mathcal{I}^{n-1}$ as $IMSPE_{sck}$ is symmetric on $\mathcal{I}^{n-1}$, that is $\mathcal{R}_{2}$ is permutation invariant in $d_{i}$. If we can show $	\dfrac{\partial \mathcal{R}_{2}(\xibold)}{\partial d_{l}} -  \dfrac{\partial \mathcal{R}_{2}(\xibold)}{\partial d_{k}} \geq 0  $, for any $d_{l} \geq d_{k}$, where $k,l = 1,\ldots,n-1$, then as before in Theorem~\ref{Theorem_SK_IMSPE} using the Schur-convexity of  $\mathcal{R}_{2}$ we will prove the equispaced design is optimal.	\\			  
		Let $q_{1}(\theta, \; \xibold)   = \{ 1 - \dfrac{n-1}{\theta} + 2 \; \Phi(\xibold)  \}$, then $\mathcal{R}_{2}(\xibold) = \int\limits_{0}^{\infty} \int\limits_{\theta_{1}}^{\theta_{2}} 
		\sigma_{11}  \;q_{1}(\theta, \; \xibold) \;\;
		r(\theta) \; t(\sigma_{11}) \;\;
		d(\sigma_{11})\;d(\theta) $.	
		Consider, 				
		\begin{align}
		%	 \text{Let, } q_{1}(\theta, \; \xi)   &= \{ 1 - \dfrac{n-1}{\theta} + 2 \Phi(\xi)  \}. \nonumber \\
		%	  \text{Then, } \mathcal{R}_{1}(\xi) &= \int_{0}^{\infty} \int_{\theta_{1}}^{\theta_{2}} 
		%				 \sigma_{11}  \;q_{1}(\theta, \; \xi) \;\;
		%				 r(\theta) \; t(\sigma_{11}) \;\;
		%				 d(\sigma_{11})\;d(\theta).	\nonumber \\
		%					\text{Consider, } 
		\Delta &=
		\dfrac{\partial \mathcal{R}_{2}(\xibold)}{\partial d_{l}} 
		-  \dfrac{\partial \mathcal{R}_{2}(\xibold)}{\partial d_{k}} \label{IMSPEPseudoBayesRiskMinimiztionEquation} \\
		&=  
		\dfrac{\partial }{\partial d_{l}} 
		\int\limits_{0}^{\infty} \int\limits_{\theta_{1}}^{\theta_{2}} 
		\sigma_{11}  \;q_{1}(\theta, \; \xibold) \;\;
		r(\theta) \; t(\sigma_{11}) \;\;
		d(\sigma_{11})\;d(\theta)     \nonumber \\
		&%\;\;\;\;\;\;\;
		-
		\dfrac{\partial }{\partial d_{k}} 
		\int\limits_{0}^{\infty} \int\limits_{\theta_{1}}^{\theta_{2}} 
		\sigma_{11}  \; q_{1}(\theta, \; \xibold) \;\;
		r(\theta) \; t(\sigma_{11}) \;\;
		d(\sigma_{11})\;d(\theta)               	       \nonumber   \\
		&=  \int\limits_{0}^{\infty} \sigma_{11}   \; t(\sigma_{11}) \; d(\sigma_{11}) 
		\Bigg[ 	 
		\int\limits_{\theta_{1}}^{\theta_{2}}  
		\Big(
		\dfrac{\partial  q_{1}(\theta, \; \xibold) }{\partial d_{l}} 
		- 	\dfrac{\partial q_{1}(\theta, \; \xibold) }{\partial d_{k}} 
		\Big) \;\; 
		r(\theta)     \;d(\theta)     	    		       
		\Bigg] \nonumber \\
		& %\;\;\;\;\;\;\;\;\;\;\;\;\;\;\;\; \;\;\;\;\;\;\;\; 
		(\text{         Using Leibniz's Rule as in \citet[chapter~8]{Book_Protter2012intermediate}}, \nonumber \\
			& \text{which allows changing the order of differentiation and integration}) \nonumber \\
		&=  
		E_{t}[\sigma_{11}] \;
		\Bigg( 	 
		2
		\int\limits_{\theta_{1}}^{\theta_{2}}  
		\Big(
		\dfrac{\partial  \Phi(\xibold) }{\partial d_{l}} 
		- 	\dfrac{\partial  \Phi(\xibold) }{\partial d_{k}} 
		\Big) \; 
		r(\theta)     \;d(\theta)     	    		       
		\Bigg) \nonumber \\
		&=   
		E_{t}[\sigma_{11}] \;
		\Bigg(
		2 	 
		\int\limits_{\theta_{1}}^{\theta_{2}}  
		\Big(
		\dfrac{\partial \phi(d_{l}) }{\partial d_{l}} \; 
		- 	\dfrac{\partial \phi(d_{k})}{\partial d_{k}} \; 
		\Big) \;\; 
		r(\theta)     \;d(\theta)     	    		       
		\Bigg). \nonumber 
		\end{align}
		
		For $d_{l} \geq d_{k}$, the quantity $\Delta$ in \eqref{IMSPEPseudoBayesRiskMinimiztionEquation} is positive, since from \eqref{SchurConvex_SimpleKriging_III} we have $\dfrac{\partial \phi(d_{l}) }{\partial d_{l}} - \dfrac{\partial \phi(d_{k}) }{\partial d_{k}} \; > 0 $ for any $d_{l} > d_{k}$. Thus, $\mathcal{R}_{2}(\xibold)$ is Schur-convex and is minimized for an equispaced design.
	\end{proof}
		Thus, we have proved the equispaced design is both locally and Bayesian optimal with respect to the $SMSPE$ and $IMSPE$ criteria for simple cokriging models. Note, for the Bayesian designs we have assumed prior distribution of covariance parameter $\theta$ with bounded support not containing zero. So, our results are true for any prior of $\theta$ with support as mentioned before.

	\subsection{Optimal designs for reduced bivariate ordinary cokriging models} \label{OptimalDesign_OrdinaryKrig}
	%	\subfile{Sections/6_Optimal_Design_Ordinary_Cokriging}
In this section, we discuss optimal designs for ordinary cokriging models with exponential covariance structures. The mean of the random function $Z_{1}(\cdot)$ is assumed to be unknown and constant (for details see Section \ref{SectionLabelOrdinaryCoKriging}). Taking a similar approach as before, in this section, we prove in Theorem~\ref{Theorem_OK_SMSPE} that the equispaced design is the G-optimal design. Though it has already been shown by \citet{Kriging_Antognini_Zagoraiou_2010} that for kriging models with unknown trend and known covariance parameter an equispaced design is I-optimal, we state the same result in Theorem~\ref{Theorem_OK_IMSPE}, since we provide an alternative way of calculating $MSPE_{ock}(x_{0})$ with simpler matrix calculations, which could be useful in the future. Also, in Theorems~\ref{Theorem_OK_SMSPE_Unknown_Parameter} and \ref{Theorem_OK_IMSPE_Unknown_Parameter} we again are able to show that the equispaced design is both locally and Bayesian I- and G-optimal.  %for covariance parameters being unknown., and hence the design is robust robust. 
	\begin{theorem} \label{Theorem_OK_SMSPE}
		Consider the reduced bivariate ordinary cokriging models as in Lemma~\ref{LemmaSimpleAndOrdinary}, where the covariance parameters, $\theta$ and $\sigma_{11}$, are known. An equispaced design minimizes the $SMSPE_{ock}$. Thus, the equispaced design is the G-optimal design.
	\end{theorem}
	\begin{proof}
		We calculate $\displaystyle{\sup_{x_0 \in [0,1]}} MSPE_{ock}(x_{0})$ and minimize it with respect to $\xibold$. From Lemma~\ref{LemmaSimpleAndOrdinary} we have,
		\begin{align*}
		SMSPE_{ock} &= \sup_{x_0 \in [0,1]} MSPE_{ock}(x_{0}) \nonumber	\\
		     &= \smash{\displaystyle \max_{i=1,\ldots,n-1} } 	\sup_{x_0 \in [x_{i},x_{i+1}]} MSPE_{ock}(x_{0}) \nonumber \\
		&=  \sigma_{11}  \smash{\displaystyle \max_{i=1,\ldots,n-1} } \sup_{x_0 \in [x_{i},x_{i+1}]}
		\Bigg( 1 - \sigmaPNot^{T} \MatrixP^{-1}\sigmaPNot   
		+ \dfrac{1}{\qZero} \Big( 1 -  \VectOneN^{T} \MatrixP^{-1} \sigmaPNot \Big)^{2}
		\Bigg). \nonumber
		\end{align*}
		From \ref{AppendixE} and \ref{AppendixF}, we can say that $\smash{\displaystyle \sup_{x_0 \in [x_{i},x_{i+1}]} } \Big( 1 - \sigmaPNot^{T} \MatrixP^{-1}\sigmaPNot \Big) $ and $\smash{\displaystyle \sup_{x_0 \in [x_{i},x_{i+1}]} } \Big( 1 -  \VectOneN^{T} \MatrixP^{-1} \sigmaPNot \Big)^{2} $ are attained at $x_{0} = x_{i} + \dfrac{d_{i}}{2}$, which is the mid-point of the interval $[x_{i},x_{i+1}]$. Also, from \ref{AppendixF} equation \eqref{EqA} we have 
		\begin{align*}
		\sup_{x_0 \in [x_{i},x_{i+1}]} \Big( 1 -  \VectOneN^{T} \MatrixP^{-1} \sigmaPNot \Big)^{2} 
		%&= U_{i}( \dfrac{d_{i}}{2}) \\
		&= \Big( 1- \dfrac{2e^{- \theta	d_{i}/2}}{1+ e^{- \theta	d_{i}}}\Big)^2 
		\end{align*}		
		Define $U_{{sup}}(\cdot)$ on $[0,1]$ such that $U_{{sup}}(d) =  \Big( 1- \dfrac{2e^{- \theta	d/2 }}{1+ e^{- \theta	d}}\Big)^2$, then $U_{{sup}}(\cdot)$ is an increasing function in $d$ as $U_{{sup}}^{\prime}(d) = 2 \theta e^{-\theta d/2} \dfrac{(1-e^{-\theta d/2})^2 (1- e^{-\theta d})}{(1+e^{-\theta d})^3} >0 $.
		
		Usually, suprema are not additive. However, if two functions $f_{1}, f_{2}: \mathcal{D}_{1} \mapsto \mathcal{D}_{2} $, where $\mathcal{D}_{1}, \mathcal{D}_{2} \subseteq \mathbb{R} $, both attain their suprema at the same point $x_{1} \in \mathcal{D}_{1}$, then we will have $\sup_{x \in \mathcal{D}_{1}} f_{1}(x) + f_{2}(x) = \sup_{x \in \mathcal{D}_{1}} f_{1}(x) + \sup_{x \in \mathcal{D}_{1}}  f_{2}(x)$.
		%We proved above that $\sup_{x_0 \in [x_{i},x_{i+1}]} \Big( 1 -  \VectOneN^{T} \MatrixP^{-1} \sigmaPNot \Big)^{2} $ is attained at  $x_{0} = x_{i} + \dfrac{d_{i}}{2}$ and is equal to $U_{{sup}}(d_{i})$. From the proof of Theorem~\ref{Theorem_SK_SMSPE} we have already seen $\sup_{x_0 \in [x_{i},x_{i+1}]} \Big( 1 - \sigmaPNot^{T} \MatrixP^{-1}\sigmaPNot  \Big)$ is attained at $x_{0} = x_{i} + \dfrac{d_{i}}{2}$. Thus,
		Thus, we write,
		\begin{align}
		\sup_{x_0 \in [x_{i},x_{i+1}]} MSPE_{ock}(x_{0}) &=  \sigma_{11}  \sup_{x_0 \in [x_{i},x_{i+1}]}
		\Bigg( 1 - \sigmaPNot^{T} \MatrixP^{-1}\sigmaPNot   
		+ \dfrac{1}{\qZero} \Big( 1 -  \VectOneN^{T} \MatrixP^{-1} \sigmaPNot \Big)^{2}
		\Bigg) \nonumber	\\  
		&=  \sigma_{11} \Bigg(  \sup_{x_0 \in [x_{i},x_{i+1}]} \Big( 1 - \sigmaPNot^{T} \MatrixP^{-1}\sigmaPNot  \Big) + \dfrac{1}{\qZero} \sup_{x_0 \in [x_{i},x_{i+1}]} \Big( 1 -  \VectOneN^{T} \MatrixP^{-1} \sigmaPNot \Big)^{2}
		\Bigg) \nonumber	\\ 	 
		&= \sigma_{11} \Big(  W_{{sup}}(d_{i}) +  \dfrac{U_{{sup}}(d_{i})}{\qZero} \Big). \label{Thm_ock_eq5}	
		\end{align}
		Hence,
		\begin{align} SMSPE_{ock} &= \sigma_{11} \; \smash{\displaystyle \max_{i=1,\ldots,n-1} } \Big(  W_{{sup}}(d_{i}) +  \dfrac{U_{{sup}}(d_{i})}{\qZero} \Big) \nonumber \\
		%\end{align}
		%As, $\Fxi$ permutation invariant, 		
		%	\begin{align}
		&= \sigma_{11} \; \Big(  W_{{sup}}(\max_i \; d_{i}) +  \dfrac{U_{{sup}}(\max_{i} \; d_{i})}{\qZero} \Big),   \label{Thm_ock_eq6} %\\
		%& (\text{As, $\Fxi$ permutation invariant}) \nonumber
		\end{align}
		as $W_{{sup}}(\cdot)$ and $U_{{sup}}(\cdot)$ are increasing functions and $\Fxi$ is permutation invariant. Since, $\smash{\displaystyle \max_{i} } \;  d_{i}$ is minimized for an equispaced partition, $W_{{sup}}( \smash{\displaystyle \max_{i} } \; d_{i})$ and $U_{{sup}}(\smash{\displaystyle \max_{i} } \; d_{i})$ are minimized for an equispaced partition. Also, $\dfrac{1}{\qZero}$ is minimized for an equispaced partition (\ref{AppendixA}). So, we have proved that the equispaced design for known $\theta$ and $\sigma_{11}$, minimizes $SMSPE_{ock}$ and therefore is G-optimal. 
	\end{proof}

	\begin{theorem} \label{Theorem_OK_IMSPE}
	Consider the reduced bivariate ordinary cokriging models as in Lemma~\ref{LemmaSimpleAndOrdinary}, with covariance parameters of the primary response, $\theta$ and $\sigma_{11}$, being known. An equispaced design minimizes the $IMSPE_{ock}$. Thus, the equispaced design is the I-optimal design.
	\end{theorem}
	\begin{proof}
				This result has been derived and proved in Theorem 4.2 by \cite{Kriging_Antognini_Zagoraiou_2010}. However, we still derive $IMSPE_{ock}$ in this paper, as we have used a different matrix approach for calculating $IMSPE_{ock}$. The approach used here is much simpler. % and, we state the final expressions as they would be needed to prove Theorem~\ref{Theorem_OK_IMSPE_Unknown_Parameter}, which supports the robustness of this design.				
		Consider a point $x_{0} \in \mathcal{D}$ and $x_0 \in [x_i, x_{i+1}], $ for some $i = 1,...,n-1$, then from Lemma~\ref{LemmaSimpleAndOrdinary}, 
		\begin{align*}
		MSPE_{ock}(x_{0}) 
		& = \sigma_{11}\Bigg[ \dfrac{\left( 1   - e^{-2 \theta a}\right) \left( 1   - e^{-2 \theta (d_{i} - a)}\right)}{\left( 1   - e^{-2 \theta d_{i}}\right)} + \dfrac{1}{\qZero} \Big(1 -  \dfrac{e^{- \theta a } + e^{- \theta( d_{i} - a )  } }{ 1 + e^{- \theta d_{i}  }} \Big)^2 \Bigg], 
		\end{align*}
		Using, \begin{align}
		IMPSE_{ock} &= \int\limits_{0}^{1}   MSPE_{ock}(x_{0}) d(x_{0}) \nonumber \\
		           & = \sum_{i=1}^{n-1} \int\limits_{x_{i} }^{x_{i+1}}   MSPE_{ock}(x_{0}) d(x_{0}) \nonumber \\
			           & = 	\sigma_{11} \sum_{i=1}^{n-1}  \Big[ \dfrac{1 + e^{-2 \theta d_{i}}}{1 - e^{-2 \theta d_{i}}}d_{i} -\dfrac{1}{\theta} \Big]   + \dfrac{\sigma_{11}}{\qZero} \sum_{i=1}^{n-1}\nonumber 
				               \Big[ d_{i} + \dfrac{-3 (1 - e^{-2 \theta d_{i}}) + 2 d_{i} \theta e^{-\theta d_{i}}}{\theta	 (1 + e^{-  \theta d_{i}})^2}\Big] \nonumber 
		\end{align}
		After doing some careful calculations, we obtain the expression for $IMSPE_{ock}$. 
		\begin{align}
		IMSPE_{ock} &= \sigma_{11} 
		\Big(
		1 - \dfrac{n-1}{\theta} + 2 \Phi(\xibold) + \dfrac{G(\xibold)}{\Fxi }
		\Big) , \label{IMSPE_OCK_Closedform}
		\end{align}
		where		
		\begin{align}
		\Phi(\xibold) &= \sum_{i=1}^{n-1} \phi(d_{i}), \;\;\;\;  \phi(d) = \dfrac{d}{e^{2 \theta d} - 1}, \nonumber \\
		G(\xibold) &= \sum_{i=1}^{n-1} g(d_{i}), \;\;\;\;  g(d) = d + \dfrac{3 (1- e^{2 \theta  d}) + 2 \theta d e^{\theta d} }{\theta	 (1 + e^{\theta d} )^2}, \nonumber \\
		\Fxi &= \sum_{i=1}^{n-1} \omega(d_{i}), \;\;\;\;  \omega(d) = d + \dfrac{e^{\theta d } - 1}{e^{\theta d} + 1}. \nonumber 	
		\end{align}
		Now using similar steps as in Theorem 4.2 of \cite{Kriging_Antognini_Zagoraiou_2010}, it can be shown that $IMSPE_{ock}$ is I-optimal.%minimized for an equispaced design. We mention the results again for Hence, equispaced design is   	
	\end{proof}
	Theorems~\ref{Theorem_OK_SMSPE} and \ref{Theorem_OK_IMSPE} both deal with the scenario in which the covariance parameters are known. To address the situation of unknown covariance parameters, we take a similar approach as in Section~\ref{OptimalDesign_unKnownParameter}. The prior distributions of $\theta$ and $\sigma_{11}$ are assumed to be known. We minimize the expected value of $SMSPE_{ock}$ and $IMSPE_{ock}$ of ordinary cokriging denoted by:
	\begin{align}
	\mathcal{R}_{3}(\xibold) & = E[SMSPE_{ock}(\theta, \sigma_{11}, \xibold )] \text{ and }  \\
	\mathcal{R}_{4}(\xibold) & = E[IMSPE_{ock}(\theta, \sigma_{11}, \xibold )].
	\end{align}  
	\begin{theorem} \label{Theorem_OK_SMSPE_Unknown_Parameter}
		Consider the reduced bivariate ordinary cokriging model as in Lemma~\ref{LemmaSimpleAndOrdinary}. The parameters $\theta$ and $\sigma_{11}$ are assumed to be unknown and independent with prior probability density functions $r(\cdot)$ and $t(\cdot)$, respectively. The support of $r(\cdot)$ is of the form $(\theta_{1},\theta_{2})$, where $0 < \theta_{1} < \theta_{2} < \infty $. Then, an equispaced design is optimal with respect to the risk function $\mathcal{R}_{3}(\xibold)$.
			\end{theorem}
	\begin{proof}				
		Denoting $\smash{\displaystyle\max_i} \;	  d_{i} = d_{max}$ we have:
		\begin{align}
		SMSPE_{ock} &= \sigma_{11} \; \Big(  W_{{sup}}(d_{max}) +  \dfrac{U_{{sup}}(d_{max})}{\qZero} \Big) \text{ from equation } \eqref{Thm_ock_eq6}. 
		\end{align}
		Let, $ q_{3}(\theta, \; \xibold)   = W_{{sup}}(d_{max}) +  \dfrac{U_{{sup}}(d_{max})}{\qZero}$.
		Then, 
		\begin{align}
		\mathcal{R}_{3}(\xibold) &= \int\limits_{0}^{\infty} \int\limits_{\theta_{1}}^{\theta_{2}} \;
		\sigma_{11}  \;q_{3}(\theta, \; \xibold) \;\;
		r(\theta) \; t(\sigma_{11}) \;
		d(\sigma_{11})\;d(\theta).	\nonumber 
		\end{align}
		Note, that $\mathcal{R}_{3}(\xibold)$ is permutation invariant of $d_{i}$'s. 
		Consider, 
		\begin{align}
		\Delta &=
		\dfrac{\partial \mathcal{R}_{3}(\xibold)}{\partial d_{l}} 
		-  \dfrac{\partial \mathcal{R}_{3}(\xibold)}{\partial d_{k}} \label{SMSPE_OCK_Unknown_Par_Eq_I}\\
		&=  
		\dfrac{\partial }{\partial d_{l}} 
		\int\limits_{0}^{\infty} \int\limits_{\theta_{1}}^{\theta_{2}} 
		\sigma_{11}  \;q_{3}(\theta, \; \xibold) \;\;
		r(\theta) \; t(\sigma_{11}) \;\;
		d(\sigma_{11})\;d(\theta)     \nonumber \\
		&\;\;\;\;\;\;\;        -
		\dfrac{\partial }{\partial d_{k}} 
		\int\limits_{0}^{\infty} \int\limits_{\theta_{1}}^{\theta_{2}} 
		\sigma_{11}  \;q_{3}(\theta, \; \xibold) \;\;
		r(\theta) \; t(\sigma_{11}) \;\;
		d(\sigma_{11})\;d(\theta)        \nonumber   \\
		&=  
		\int\limits_{0}^{\infty} \sigma_{11}   \; \mathbf{t}(\sigma_{11}) \; d(\sigma_{11}) 
		\Bigg[ 	 
		\int\limits_{\theta_{1}}^{\theta_{2}}  
		\Big(
		\dfrac{\partial  q_{3}(\theta, \; \xibold) }{\partial d_{l}} 
		- 	\dfrac{\partial  q_{3}(\theta, \; \xibold) }{\partial d_{k}} 
		\Big) \;\; 
		r(\theta)     \;\;d(\theta)     	    		       
		\Bigg] \nonumber \\
		&\;\;\;\;\;\;\;\; (\text{  Using Leibniz's Rule as in \citet[chapter~8]{Book_Protter2012intermediate}} ) \nonumber \\
		&=  
		E_{t}(\sigma_{11})  
		\Bigg[ 	 
		\int\limits_{\theta_{1}}^{\theta_{2}}  
		\Big(
		\dfrac{\partial  q_{3}(\theta, \; \xibold) }{\partial d_{l}} 
		- 	\dfrac{\partial  q_{3}(\theta, \; \xibold) }{\partial d_{k}} 
		\Big) \; 
		r(\theta)     \;d(\theta)     	    		       
		\Bigg]  . \nonumber 
		\end{align}
		Note, 
		\begin{align}
		\text{for }d_{i} \neq d_{max}, \;\; 
		\dfrac{\partial q_{3}(\theta, \; \xibold) }{\partial d_{i}}  &= - \dfrac{U_{sup}(d_{max})}{(\Fxi)^2} \dfrac{\partial \omega(d_{i}) }{\partial d_{i}}  \nonumber \\
		\text{and, if } d_{i} =  d_{max}, \;\; 
		\dfrac{\partial q_{3}(\theta, \; \xibold) }{\partial d_{i}}  
		&= W_{sup}^{\prime}(d_{max}) + \dfrac{U_{sup}^{\prime}(d_{max})}{\Fxi} - \dfrac{U_{sup}(d_{max})}{(\Fxi)^2} \dfrac{\partial \omega(d_{max}) }{\partial d_{max}} . \nonumber		    
		\end{align}
		Thus,
		\begin{equation} \label{SMSPE_OCK_Unknown_Par_Eq_II}
		\dfrac{\partial q_{3}(\theta, \; \xibold)}{\partial d_{l}}  
		- 	\dfrac{\partial q_{3}(\theta, \; \xibold)}{\partial d_{k}}  = 
		\begin{cases}
		\dfrac{U_{sup}(d_{max})}{(\Fxi)^2} \Big(\dfrac{\partial \omega(d_{k}) }{\partial d_{k}} - \dfrac{\partial \omega(d_{l}) }{\partial d_{l}} \Big)   &\text{ for } d_{k}, d_{l} \neq d_{max} \\
		& \\
		W_{sup}^{\prime}(d_{max}) + \dfrac{U_{sup}^{\prime}(d_{max})}{\Fxi} & \\
		\; + \; \dfrac{U_{sup}(d_{max})}{(\Fxi)^2} \Big(\dfrac{\partial \omega(d_{k}) }{\partial d_{k}} - \dfrac{\partial \omega(d_{max}) }{\partial d_{max}}\Big) &\text{ for } d_{k} \neq d_{l} =   d_{max} 
		\end{cases}
		\end{equation} 
		Note, that for $d_{l} > d_{k}$, 		from $\eqref{Appendixkey1}$ we have $\Big(    \dfrac{\partial \omega(d_{k}) }{\partial d_{k}} - \dfrac{\partial \omega(d_{l})}{\partial d_{l}} \Big) >  0 $ and from Theorems~\ref{Theorem_SK_SMSPE} and \ref{Theorem_OK_SMSPE}, $W_{sup}^{\prime}(.) >0$ and $U_{sup}^{\prime}(.) >0 $. Hence, the terms in equation \eqref{SMSPE_OCK_Unknown_Par_Eq_II} $>0 $. 
		
		So, from equation $\eqref{SMSPE_OCK_Unknown_Par_Eq_I}$ we get $	\dfrac{\partial \mathcal{R}_{3}(\xibold)}{\partial d_{l}} 
		-  \dfrac{\partial \mathcal{R}_{3}(\xibold)}{\partial d_{k}} >0 $ for $d_{l} > d_{k}$, which implies $\mathcal{R}_{3}(\xibold)$ is Schur-convex and is minimized for an equispaced design.
	\end{proof}
	\begin{theorem} \label{Theorem_OK_IMSPE_Unknown_Parameter}
		Consider the reduced bivariate ordinary cokriging model as in Lemma~\ref{LemmaSimpleAndOrdinary}. The parameters $\theta$ and $\sigma_{11}$ are assumed to be unknown and independent with prior probability density functions $r(\cdot)$ and $t(\cdot)$, respectively. The support of $r(\cdot)$ is of the form $(\theta_{1},\theta_{2})$, where $0 < \theta_{1} < \theta_{2} < \infty $. Then, an equispaced design is optimal with respect to the risk function $\mathcal{R}_{4}(\xibold)$.
\end{theorem}
	\begin{proof}
		Using the same line of proof as in Theorem \eqref{Theorem_SK_IMSPE_UnknownParameter} we can show that the equispaced design is I-optimal for an unknown parameter case as well.		
	\end{proof}

	\section{Case study} \label{Illustration_cokriging_known_parameter}
	In this section, we are interested in using the proposed optimality results in the earlier section to design a river monitoring network for the efficient prediction of water quality. A pilot data set of water quality data from river Neyyar in southern India is used to obtain preliminary information about parameters. We will illustrate how the theory that we developed in Sections~\ref{SectionModelReduction} and \ref{OptimalDesign_KnownParameter} is applied to this problem. The image of the river is shown in Figure \ref{fig1}, where the monitoring stations on the river basin are marked with squares. We will compare the performance of the equispaced design with the given choice of stations for designing a cokriging experiment on this river. 
	\begin{figure}[h]
		\includegraphics[height=10cm, width = 15cm, scale = 5]{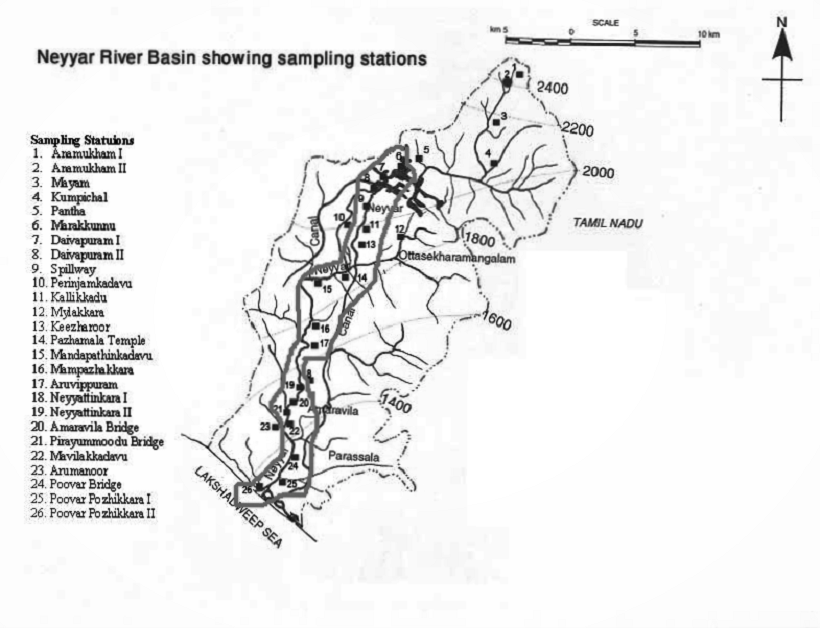}
		\caption {Monitoring station positions on the Neyyar river basin. We use the station locations and data within the encircled area.}
		\label{fig1}
	\end{figure}
	
	%The image of the river is shown in Figure 1 in Supplement II \cite{SupplementTwo} %Figure \ref{fig1} 
	%where the monitoring stations on the river basin are marked in red. We will compare the performance of the equispaced design with the given design of stations. \\
	%\begin{figure}[h]
	%	\includegraphics[height=10cm, width = 15cm, scale = 5]{ImageNeyyarSurface_2}
	%	\caption {Monitoring station postions on the Neyyar river basin. We use the station locations and data within the green area.}
	%	\label{fig1}
	%\end{figure}
	The location of each monitoring station is specified by its geographical coordinates, that is, latitude and longitude. At each of these stations, measurements are taken for two variables: pH and phosphate which are used to measure the quality of water. For carrying out the analysis, that is, gathering information on the direct covariance and cross-covariance functions and parameters of the two responses, we use data from a single branch of the river with 17 stations (see the encircled region in Figure \ref{fig1})%in Supplement II \cite{SupplementTwo}) %\ref{fig1})
	. We denote this branch of the river by $\mathcal{D}_{2} \;(\subseteq \mathbb{R}^{2})$ and in this case we have $n = 17 $.  We denote the set of sampling points on this river branch by $\mathcal{S}_{2} = \{w_{1}, \ldots,w_{17}\} \;(\subseteq \mathcal{D}_{2})$, where each $w_i=(latitude_i,longitude_i),\,i=1,\ldots,17$. Let $w_{1}$ and $w_{17}$ respectively be the starting (station 6) and the end point (station 26) of the river branch, and suppose we assume $w_{i}$ is upstream of $w_{j}$ if $i<j$ for all $i,j=1,\ldots,17$. 
	
	The results that we obtained for determining optimal designs in earlier sections were based on one-dimensional inputs, that is, where the region of interest was denoted by $\mathcal{D} \subset \mathbb{R}$. In fact, without loss of generality we had assumed $\mathcal{D} = [0,1]$. So, we first use a transformation on our two-dimensional input sets $\mathcal{S}_{2}$ and $\mathcal{D}_{2}$ given by:
	%%These random functions are recorded at the set of locations $\mathcal{S}_{1} = \{s_{1}, ...,s_{n}\}$ where each $s_i=(latitude_i,longitude_i),\,i=1,\ldots,17$. Let $s_{1}$ and $s_{n}$ be the starting (i.e., station 6) and the end point (i.e., station 26) of the river branch, and suppose $s_{i}$ is upstream of $s_{j}$ if $i<j$. 
	%%%Although, $\mathcal{S} \subseteq \mathbb{R}^2$, but use the following transformation to translate our problem in $\mathbb{R}$.
	%%		\begin{align*}
	%%		\text{We define a map: }
	%%		\varphi: \mathcal{S}_2 & \longrightarrow [0,1]\\%Although, $\mathcal{S} \subseteq \mathbb{R}^2$, but use the following transformation to translate our problem in $\mathbb{R}$.
	\begin{align*}
	%\text{Define a map: }
	\varphi: \mathcal{D}_2 & \longrightarrow [0,1]\\
	w & \mapsto \dfrac{\mid \mid w-w_{1}\mid \mid}{\mid \mid w_{17}- w_{1}\mid \mid}, 
	\end{align*}
	where $ \mid \mid u- v \mid \mid $ is the geodesic stream distance between the two points $u$ and $v$ along the river and $ u,v \in \mathcal{D}_2 $. The geodesic distance is used to calculate distance on the earth's surface and is discussed in \cite{Book_sudipto_banerjee2014} in detail. The stream distance is the shortest distance between two locations on a stream, where the distance is computed along the stream \citep{Ver_Hoef_2006}. In this case it was not possible to calculate the exact stream distance using solely the coordinates of monitoring points. So, the stream distance between two adjacent points was approximated by the geodesic distance between the two points. 
	
	The transformed region of interest $ \varphi(\mathcal{D}_2)  = \mathcal{D}_{1} = [0,1]$ and the set of sampling points $\varphi(\mathcal{S}_{2})  = \mathcal{S}_{1} $ are one-dimensional. We had to constrain ourselves to a single branch of river as a single branch of river is connected and hence can be considered to be a one-dimensional object. For example, consider stations 10, 18 and 23 which are very close to the main branch, but if these points were included, then the transformation  to a one-dimensional set would not work. The transformed set of observation points is given by $\mathcal{S}_{1} = \{ x_{1}, x_{2}, \ldots,x_{17} \}$ where $\varphi(w_{i}) = x_{i}$ for all $i=1,\ldots ,17$. 
	Also, by definition of the function $\varphi(\cdot)$ $ x_{1}=0$, $x_{17} =1 $ and $x_{i} < x_{j}$ for $ i<j $, and $d_{i}  = x_{i+1} - x_{i}$ for $i=1,\ldots,16$. 
	
	We took the pH level (a scalar with no units) as the primary variable  $Z_{1}(\cdot)$, and phosphate concentration (measured in mg/l) as the secondary variable $Z_{2}(\cdot)$, with both the variables centered and scaled. 

	To investigate the covariance function and corresponding parameters we fitted a model by likelihood maximization, separately for each variable. Below we see Table \ref{tab1}, which was computed using the $likfit$ function with a constant mean (that is, corresponding to unknown mean) from the  $geoR$ package (R-3.6.0  software). %1 in Supplement II \cite{SupplementTwo} %\ref{tab1} 
	The likelihood values in Table~\ref{tab1} suggest that taking the random processes as a zero-mean process with an exponential variance structure and zero nugget effect is a reasonable choice for both variables. Using the information from the univariate analysis of pH and phosphate we next try to set up the appropriate bivariate simple cokriging model. Note that for both variables, we tried to fit a Gaussian covariance structure, however, the algorithm did not converge.

	\begin{table}[h]
		\begin{center}

{\scriptsize
	%\footnotesize
\begin{tabular}{ m{3cm} ll lll   }
	
		\multicolumn{6}{c}{\textbf{pH}} \\
		\hline
		Covariance  Model & $C(h) = \sigma^2 \rho(h)$ &Log-Likelihood & Variance & Parameter  ($\theta$, $\kappa$) & Nugget \\
		\hline
		\multicolumn{6}{l}{Constant but unknown mean}\\
		\hline
		Exponential      & $\sigma^2 exp(-\theta |h|)$ & -20.28     &  0.85 &  16.95 & 0 \\
		%Exponential      & $\sigma^2 exp(-\theta |h|)$ & -23.61     &  0.00 &  $\infty$ &0.94 \\
		%Gaussian	     & $\sigma^2 exp(-(\theta |h|)^2)$ & No Convergence	 & NA  &NA&NA\\
		Spherical	     & $\sigma^2 
		\begin{cases}
		1 - 1.5 h \theta +  
		.5 (h \theta)^3 , &  \textit{ if }  h < \dfrac{1}{\theta}\\
		0,              & \text{otherwise}
		\end{cases}                                      
		$ & -20.74     & 0.96 &  7.90 &0\\
		Matern	         & $\sigma^2 \dfrac{1}{2^{(\kappa-1)} \Gamma (\kappa) } (h \theta)^{\kappa} K_{\kappa}(h \theta)$ & -20.15	   & 0.83 &  (11.09,0.35)&0\\
		\hline
		\multicolumn{6}{l}{Known mean equal to zero}\\
		Exponential      & $\sigma^2 exp(-\theta |h|)$ & -20.29     &  0.85 &  17.12 &0 \\
		\hline \\
		\multicolumn{6}{c}{\textbf{phosphate}} \\
		\hline
		Covariance Model & $C(h) = \sigma^2 \rho(h)$ &Log-Likelihood & Variance & Parameter  ($\theta$, $\kappa$) & Nugget\\
		\hline 
		\multicolumn{6}{l}{Constant but unknown mean}\\
		\hline
		Exponential      & $\sigma^2 exp(-\theta |h|)$ & -23.19          & 0.97 & 38.35 &0 \\
		%Exponential      & $\sigma^2 exp(-\theta |h|)$ &  -23.61          & 	0.00 & 38.31 &.94\\
		%Gaussian	     & $\sigma^2 exp(-(\theta |h|)^2)$ &  No Convergence	&  NA  & NA &NA\\
		Spherical	     & $\sigma^2 
		\begin{cases}
		1 - 1.5 h \theta +  
		.5 (h \theta)^3 , & \textit{ if } h < \dfrac{1}{\theta}\\
		0,              & \text{otherwise}
		\end{cases}                                      
		$ & -23.09         & 0.95 & 19.02 &0\\
		Matern	         & $\sigma^2 \dfrac{1}{2^{(\kappa-1)} \Gamma (\kappa) } (h \theta)^{\kappa} K_{\kappa}(h \theta)$ &  -23.85         & 0.97 & (0.01,0.003) &0\\
		\hline
		\multicolumn{6}{l}{Known mean equal to zero}\\
		Exponential& $\sigma^2 exp(-\theta |h|)$ & -23.29     &  0.96 &  45.94 &0 \\
		\hline
	\end{tabular}
}
	\end{center}
\caption{Results of Likelihood Analysis of pH and Phosphate for Different Covariance Models}
\label{tab1}
	\end{table}
		 %, ii) When we tried to fit an exponential covariance structure with unknown mean and non zero nugget using the $likfit$, the estimates are found not reliable and hence we omit it from the table. iii) For both variables, exponential covariance structure for Known mean equal to zero is calculated using calculated using $optim$ function in R.  

	We consider $Z_{1}(\cdot)$ and $Z_{2}(\cdot)$ to have the exponential parameters $\theta$ and $\phi$, respectively. The results from Table \ref{tab1} %in Supplement II \cite{SupplementTwo} %~\ref{tab1} 
	for pH and phosphate indicate a large difference between $\hat{\theta}$ and $\hat{\phi}$. Thus, it seems more appropriate to assume a generalized Markov-type bivariate covariance rather than proportional covariances in the bivariate cokriging model. %Next, we conducted a bivariate analysis to understand the structure of the data covariance structure \textit{Generalized Markov Type} as in eq.(\ref{GeneralizedMarkovMatrix}).
	Based on the assumption of normal errors, the  log-likelihood function  is: 
	\begin{equation*}l=- \nhalf  log(2 \pi) - \half log[det(\covMat)]  -  \half \VectZ^{T} \covMat^{-1} \VectZ,  
	\end{equation*} where $\VectZ = (\VectZOne ,\VectZTwo )^{T}$, $\covMat =  
	\begin{bmatrix} 
	\MatrixMOne & \rho \MatrixMOne\\
	\rho \MatrixMOne &  \rho^{2} \MatrixMOne+ ( \sigma_{22}-\rho^{2}  \sigma_{11}) \MatrixCR
	\end{bmatrix}$, and  $\MatrixCR$ is chosen to be the identity matrix.%\text{ to be the identity matrix, }$
	
	%As there is no function in R specifically to estimate parameters of generalized Markov type covariance structures, for the bivariate model, we use the $optim$ function in \textit{R Software} to find the MLEs  $\hat{\theta_{1}} = 17.12, $, $\hat{\sigma_{11}} = 0.85$, $\hat{\sigma_{22}} = 0.94,$ and $\hat{\rho} = .25$ with the log-likelihood value - 27.74. 
	Using the $optim$ function in (R-3.6.0 software) we find the MLEs  to be $\hat{\theta} = 17.12$, $\hat{\sigma_{11}} = 0.85$, $\hat{\sigma_{22}} = 0.94,$ $\hat{\rho} = .25$ and $l= - 27.74$. The $likfit$ and $optim$ functions in R-3.6.0 were used for computations.
	%Next we compute the  efficiency of the river network with respect to the  equispaced optimal design. \\
	%	The distance $d_i$'s in the river  monitoring stations are given by the set 
	% In upcoming illustrations we define the efficiency of design for the two design criterion. We also show an equispaced design has advantage over the given design. 
	
	\begin{illustration}
		Relative efficiency when parameter values are known 
	\end{illustration}
	The design given for the pilot monitoring network is denoted by $\xiboldNot$, which is obtained by considering the 17 points on the river (encircled region) and applying the transformation $\varphi(\cdot)$. We computed  $\xiboldNot 
	= (
	0.04,
	0.02,
	0.04,
	0.09,
	0.20,
	0.06,
	0.12,
	0.13,
	0.04,
	0.04,
	0.02, 
	0.05,
	0.04,\\
	0.07,
	0.02,
	0.02
	)$. 
	%	$\xi_0= 
	%		\{0.00,
	%		0.06, 
	%		0.10, 
	%		0.19, 
	%		0.39,$
	%		$0.45, 
	%		0.57, 
	%		0.70, 
	%		0.74, 
	%		0.78, 
	%		0.80, 
	%		0.85, 
	%		0.90, 
	%		0.97, 
	%		0.98, 
	%		1.00\}$.	 
	We also denoted the equispaced design by $\xiboldAst$, where $(\xiboldAst)_{i} = \dfrac{1}{n-1}  = \dfrac{1}{16}$ for all $i = 1,\ldots,17$. The parameter values are taken to be the same as the maximum likelihood estimates. 
	
	Relative efficiency based on IMSPE of design $\xiboldNot$ with respect to the optimal design $\xiboldAst$ is defined as the ratio,
	$\dfrac{ \text{IMSPE($\xiboldAst$)}}{\text{IMSPE($\xiboldNot$)}}$.		
	%		We have seen in $Theorem. \ref{IMSPESimpleCokriginTheoremBothCase}$ that for simple cokriging the $IMSPE$ depends on the parameters of primary variable $\theta_{1}$ and $\sigma_{11}$.    
	For known parameters, using the expression of IMSPE in Theorem~\ref{Theorem_SK_IMSPE}, the relative efficiency of the river network (or design) $\xiboldNot$ is found to be 0.797. Similarly, for the SMSPE criterion we define the ratio as $\dfrac{ \text{SMSPE($\xiboldAst$)}}{\text{SMSPE($\xiboldNot$)}}$. For the SMSPE criterion, using Theorem~\ref{Theorem_SK_SMSPE} the relative efficiency of the river network $\xiboldNot$ is 0.524. Note that relative efficiency values in both cases indicate a sizable increase in prediction accuracy if equispaced designs were used instead.		   
	\begin{illustration}
		Relative efficiency  for unknown parameters 
	\end{illustration}
Consider, $\theta \sim Unif(\theta_{1}, \theta_{2})$ for $0< \theta_{1} <  \theta_{2} < \infty$, a common choice of prior for $\theta$ (see \cite{stehlik2015robust}) %referred to as the kernels
and $\sigma_{11} \sim t(\cdot)$ for some density function $t(\cdot)$.  Note we could have chosen any prior function for $\theta$ other than the uniform distribution as long as it had a finite support. The risks for the uniform prior are,%$\mathcal{R}_{1}(\xi)$ is given by:
	\begin{align}
	\mathcal{R}_{1}(\xibold) & =
	E_{\sigma} \;\; \dfrac{1}{\theta_{2} - \theta_{1}} \dfrac{1}{d_{max}}\Bigg[ 2 ln  \dfrac{1 + e^{- \theta_{2} d_{max}} }{ 1 + e^{- \theta_{1} d_{max}} } + d_{max}(\theta_{2} - \theta_{1})  \Bigg] \label{risk1}
	\end{align}
	and,
	\begin{align}
	\mathcal{R}_{2}(\xibold) & =
	E_{\sigma} \;\; \Bigg[1 - \dfrac{n-1}{\theta_{2} - \theta_{1}} \;\; ln\dfrac{\theta_{2}}{\theta_{1}} + \dfrac{1}{\theta_{2} - \theta_{1}} \sum_{i=1}^{n-1} ln \Big( \dfrac{e^{2 \theta_{2} d_{i}} -1}{e^{2 \theta_{2} d_{i}} } . 
	\dfrac{e^{2 \theta_{1} d_{i}} }{e^{2 \theta_{1} d_{i}} - 1 } \Big) \Bigg] \label{risk2},
	\end{align}
	where $\smash{\displaystyle\max_i}	 (d_i)$ is written as $d_{max}$ and $E_{\sigma}=E_{t}[\sigma_{11}]$. The relative efficiency is then $\dfrac{{\mathcal{R}_i(\xiboldAst)}}{{\mathcal{R}_i(\xiboldNot)}},\,i=1,2$.	Note, these risks in \eqref{risk1} and \eqref{risk2} would differ if we change the prior. However $\xiboldAst$ would remain same.  
	
	Using $\hat{\theta} = 17.12$, we choose $\theta_1$ and $\theta_2$ such that the mean of the interval is  $\hat{\theta}$. Varying the range of values for $\theta_1$ and $\theta_2$, the relative risks are shown in the following Table \ref{Table_Relative Risk}.
	%Table 2 of Supplement II \cite{SupplementTwo}. 
	\begin{table}[h]
	\begin{center}

		\begin{tabular}{ @{}lrrrrcll@{} }
					\hline
			$\theta_{1}$  & $\theta_{2}$  &  $\mathcal{R}_1(\xiboldAst)/E_{\sigma}$ & $\mathcal{R}_1(\xiboldNot)/E_{\sigma}$  & $\dfrac{ \mathcal{R}_1(\xiboldAst)}{\mathcal{R}_1(\xiboldNot)}$ &
			$\mathcal{R}_2(\xiboldAst)/E_{\sigma}$ & $\mathcal{R}_2(\xiboldNot)/E_{\sigma}$  & $\dfrac{ \mathcal{R}_2(\xiboldAst)}{\mathcal{R}_2(\xiboldNot)}$ \\
			\hline
			16.62	&	17.62	&	0.489	&	0.933	&	0.524 &	0.332	&	0.434	&	0.766	\\
			16.12	&	18.12	&	0.489	&	0.933	&	0.524&	0.332	&	0.433	&	0.766	\\
			15.12	&	19.12	&	0.489	&	0.932	&	0.525&	0.332	&	0.433	&	0.766	\\
			12.12	&	22.12	&	0.486	&	0.923	&	0.527&	0.330	&	0.430	&	0.768	\\
			\hline
		\end{tabular}
		\end{center}

		\caption{Relative risk of given design - IMSPE and SMSPE criterion}
			\label{Table_Relative Risk}	
	\end{table}
	From Table %2 \cite{SupplementTwo} %Table
	\ref{Table_Relative Risk}, we note small changes in the relative efficiency for changes in $\theta_1$ and $\theta_2$, suggesting that the criterion is robust to changes in the prior information regarding $\theta$. This robustness persists when we change the values of $\hat{\theta}$. We also checked relative efficiencies for $\hat{\theta} = $ 7.12, 27.12 and 47.12, however the results are not shown here. %. Due to constrain of space we do not add those tables. 
	
		\section{Concluding remarks} \label{SectionConcludingRemarks}
Multivariate kriging models are of particular practical interest in computer experiments, spatial and spatio-temporal applications. Very often, two or more correlated responses may be observed, and prediction from cokriging may improve prediction quality over kriging for each variable separately.

In this article, we identify a class of cross-covariance functions, which in fact includes many popularly used bivariate covariance functions, for which the cokriging estimator reduces to a kriging estimator. Thereafter, we address the problem of determining designs for some of these cokriging models. Since the designs are dependent on the covariance parameters, Bayesian designs are proposed. We prove that the locally and Bayesian optimal designs are both equispaced. Intuitively, this could be explained due to the fact that the locally optimal designs are equispaced for all the values of covariance parameters. So, when we mathematically find the Bayesian optimal designs, both are equispaced. 

As a future extension, we are interested in determining optimal designs for universal cokriging models. However, as illustrated in \cite{Dette_et_al_2008} and \cite{Dette_et_al_2013}, obtaining theoretical designs for such models is difficult. We have also come across situations in cokriging experiments where time and space (or multiple inputs) both may affect the responses. Thus, there is a need to extend optimal designs to cover such scenarios where the input space is a multidimensional grid of points. 	
%	The main results obtained are summarized below:
%	\begin{itemize}
%		\item For a class of cross covariance structures, collocated cokriging models are reduced to kriging models
%		\item A new method to calculate $MSPE_{sk}(\cdot)$ and $MSPE_{ok}(\cdot)$
%		\item Equispaced design minimizes the SMSPE and IMSPE for simple and ordinary kriging models with exponential covariance structure, when covariance parameters are assumed to be known.
%		\item Equispaced design minimizes the SMSPE and IMSPE for simple and ordinary kriging models with exponential covariance structure, when prior distributions are assumed on the covariance parameters. 
%	\end{itemize} 
	
\begin{appendix} 
			\section{Appendix} \label{AppendixCov} 
		\begin{result}\label{GenMarkovValidity}
		Consider two random functions $Z_{1}(\cdot)$ and $Z_{2}(\cdot)$ with respective covariance functions $\mathcal{C}_{ii}(\cdot)$ and spectral densities  $s_{i}(\cdot)$ for $i =1,2$. Consider another valid correlation function $\CovFunctR(\cdot)$ with spectral density $s_{R}(\cdot)$. Then, $ \covMat$ as defined in \eqref{GeneralizedMarkovMatrix} is a valid covariance matrix if and only if $( \sigma_{22}-\rho^{2}  \sigma_{11}) > 0$.
	\end{result}
	
	\begin{proof}
		
		The cross-spectral density matrix $\pmb {S}_p(u)$ is,
		\begin{align*}
		\pmb{S}_p(u) &= 
		\begin{bmatrix} 
		s_{1}(u) & \rho s_{1}(u)\\
		\rho s_{1}(u) & \;\; \rho^{2} s_{1}(u)+ ( \sigma_{22}-\rho^{2}  \sigma_{11}) s_{R}(u)
		\end{bmatrix}, u \in \mathbb{R}
		\end{align*}
		with determinant $s_1(u)( \sigma_{22}-\rho^{2}  \sigma_{11})s_R(u)$. Note, that the matrix $\pmb{S}_{p}(u)$ is positive definite whenever $( \sigma_{22}-\rho^{2}  \sigma_{11}) > 0$,  as $s_{1}(\cdot)$ and $s_{R}(\cdot)$ correspond to the inverse Fourier transforms of the covariance functions $\CovFunctOneOne(\cdot)$ and $\CovFunctR(\cdot)$, respectively. Using the criterion of \cite{Cramer_1940_theory}, $ \covMat$ is then a valid covariance matrix if and only if $( \sigma_{22}-\rho^{2}  \sigma_{11}) > 0$.
	\end{proof}

		\section{Appendix} \label{Appendix} 
		We list down some of the key matrices, vectors and their decomposition required for proving results in Lemma~\ref{LemmaSimpleAndOrdinary} and Theorems \ref{Theorem_SK_SMSPE}, \ref{Theorem_SK_IMSPE}, \ref{Theorem_OK_SMSPE} and \ref{Theorem_OK_IMSPE}. In this article, we have used an exponential covariance matrix $\MatrixP$, where 
		\begin{align*}
		\MatrixP &= 
		\begin{bmatrix}
		1 & e^{- \theta |x_{1}-x_{2}|} &\;\;\;\;\;\;\; &\ldots& e^{- \theta |x_{1}-x_{n}|}  \\
		e^{- \theta |x_{2}-x_{1}|} & 1 &\;\;\;\;\;\;\; &\ldots& e^{- \theta |x_{2}-x_{n}|} \\
		. & .                  &\;\;\;\;\;\;\; &\ldots& .									  \\
		. & .                  &\;\;\;\;\;\;\; &\ldots& .									  \\
		. & .                  &\;\;\;\;\;\;\; &\ldots& .									  \\
		e^{- \theta |x_{n}-x_{1}|}  & e^{- \theta |x_{n}-x_{2}|}  &\;\;\;\;\;\;\; &\ldots& 1
		\end{bmatrix}.
		\end{align*}
		%%% Clipped for length sake
		Considering matrices $\MatrixL$ and $\MatrixD$, as in \cite{Kriging_Antognini_Zagoraiou_2010}, 
		\begin{align*}
		\MatrixL &= 
		\begin{bmatrix}
		1 & 0 & 0 & .&.&.& 0\\
		e^{-\theta d_{1}} & 1 & 0 & .&.&.&0\\
		e^{-\theta \Sigma^{2}_{i=1}  d_{i} } & e^{-\theta  d_{2} } & 1 & .&.&.&0\\
		. & . & . &  .&.& &\\
		. & . & . & .&..&&\\
		e^{-\theta \Sigma^{n-1}_{i=1} d_{i} } & e^{-\theta \Sigma^{n-1}_{i=2} d_{i} } & e^{-\theta \Sigma^{n-1}_{i=3} d_{i} } & .&.&.&1\\
		\end{bmatrix} \text{ and } \MatrixD = diag(1,1-e^{-2\theta d_{1}},\ldots, 1-e^{-2\theta d_{n-1}}),
		\end{align*}
		we wrote $
		\MatrixP = \MatrixL \MatrixD \MatrixL^{T}$. Thus, 
		  %		\begin{align}
%		\MatrixP &= \MatrixL \MatrixD \MatrixL^{T}. \nonumber	 
%		\end{align}
		\begin{align}
		\MatrixP^{-1} &= (\MatrixD^{-1/2} \MatrixL^{-1})^{T} (\MatrixD^{-1/2} \MatrixL^{-1}), \label{CovMatDecmp}
		\end{align}
where 
		\begin{align}
		\MatrixP^{-1}	= & 
		\begin{bmatrix}
		\dfrac{1}{1-e^{- 2 \theta d_{1}}} & \dfrac{- e^{\theta d_{1}}}{1-e^{- 2 \theta d_{1}}} &     & . & .& .& 0 \\ 
		\dfrac{- e^{\theta d_{1}}}{1-e^{- 2 \theta d_{1}}} &  \dfrac{1}{1-e^{- 2 \theta d_{1}}} + \dfrac{ e^{-2 \theta d_{2}}}{1-e^{- 2 \theta d_{2}}} & & 0 & .& .& 0 \\
		.\\
		.\\
		.\\
		0 & . &  & 0 & . & \dfrac{1}{1-e^{- 2 \theta d_{n-2}}} + \dfrac{ e^{-2 \theta d_{n-1}}}{1-e^{- 2 \theta d_{n-1}}} & \dfrac{- e^{\theta d_{n-1}}}{1-e^{- 2 \theta d_{n-1}}} \\
		0&0& &.&.& \dfrac{- e^{\theta d_{n-1}}}{1-e^{- 2 \theta d_{n-1}}}  & \dfrac{1}{1-e^{- 2 \theta d_{n-1}}} 
		\end{bmatrix} \label{CovMatInverse} .
		\end{align}
		%For matrices, 

		\section{Appendix} \label{Appendix_IMSPE_SCKandOCK} \label{AppendixA}
		\noindent Here, we evaluate $\Fxi = \VectOneN^{T} \MatrixP^{-1} \VectOneN $ and show that $\dfrac{1}{\qZero}$ is a Schur-convex function, which is minimized for an equispaced partition. Using equation \eqref{CovMatDecmp} from \ref{Appendix}, we write, 
		\begin{align}
		\VectOneN^{T} \MatrixP^{-1} \VectOneN &= (\MatrixD^{-1/2} \MatrixL^{-1} \VectOneN)^{T} (\MatrixD^{-1/2} \MatrixL^{-1} \VectOneN) = \VectGamma^{T} \VectGamma, \nonumber 
		\end{align}
		where
		\begin{align}
		\VectGamma^{T} &=	( \MatrixD^{-1/2} \MatrixL^{-1} \VectOneN)^{T} \;\;= \Big(1, \dfrac{1-e^{-\theta d_{1}}}{\sqrt{(1-e^{-2\theta d_{1}})}}, \ldots, \dfrac{1-e^{-\theta d_{n-1}}}{\sqrt{(1-e^{-2\theta d_{n-1}})}} \Big).  \nonumber 
		\end{align}
		Hence we have, 
%		\begin{align}
%		\VectOneN^{T} \MatrixP^{-1} \VectOneN  &= 1 +  \sum_{i=1}^{n-1} \dfrac{e^{\theta d_{i}} - 1}{e^{\theta d_{i}} + 1 } \nonumber \\
%		&= \sum_{i=1}^{n-1} \Big[ d_{i} +  \dfrac{e^{\theta d_{i}} - 1}{e^{\theta d_{i}} + 1 }\Big] \;\;\;\;\;\;\;\;\; \text{      (Note that withoiut loss of generality we assumed $\sum_{i=1}^{n-1} d_{i}=1$). } \nonumber
%		\end{align}
		\begin{align}
		\VectOneN^{T} \MatrixP^{-1} \VectOneN  &= 1 +  \sum_{i=1}^{n-1} \dfrac{e^{\theta d_{i}} - 1}{e^{\theta d_{i}} + 1 }.  \nonumber 
		\end{align}
		As without loss of generality we assumed $\sum_{i=1}^{n-1} d_{i}=1$, therefore 
				\begin{align}
		\VectOneN^{T} \MatrixP^{-1} \VectOneN  
		&= \sum_{i=1}^{n-1} \Big[ d_{i} +  \dfrac{e^{\theta d_{i}} - 1}{e^{\theta d_{i}} + 1 }\Big]  .  \nonumber
		\end{align}
		Using the above expression we write  
		\begin{align}
		\Fxi &= \sum_{i=1}^{n-1} \omega(d_{i}), \text{ where, } \omega(d) = d +  \dfrac{e^{\theta d} - 1}{e^{\theta d} + 1 }. \label{F_xi_term}
		\end{align}
		%\begin{align}
		%\text{Consider, } Q(\xi) =& \dfrac{1}{\Fxi} \nonumber \\
		%\text{First see that, }
		%\dfrac{\partial f(\xi)}{\partial d_{i}} & = 1 + \dfrac{2 \theta e^{\theta d_{i}}}{(1 + e^{\theta d_{i}})^2} \nonumber \\
		%\dfrac{\partial^{2} f(\xi)}{\partial d_{i}^{2}} & =  \dfrac{2 \theta^{2} e^{\theta d_{i}} (1-e^{\theta d_{i}})}{(1 + e^{\theta d_{i}})^3} < 0  \label{Appendixkey1} \\
		%\textit{And, } \dfrac{\partial Q(\xi)}{\partial d_{2}} - \dfrac{\partial Q(\xi)}{\partial d_{1}} & = \dfrac{1}{(F(\xi))^{2}} \Big[ \dfrac{\partial f(d_{1})}{\partial d_{1}} - \dfrac{\partial f(d_{2})}{\partial d_{2}}  \Big] \label{Appendixkey2} 
		%\end{align} 
		Next, differentiating $\Fxi$ with respect to $d_{i}$ we obtain
		\begin{align}
		\dfrac{\partial \Fxi}{\partial d_{i}} & = 1 + \dfrac{2 \theta e^{\theta d_{i}}}{(e^{\theta d_{i}} + 1  )^2}, \nonumber \\
		\dfrac{\partial^{2} \Fxi}{\partial d_{i}^{2}} & =  \dfrac{2 \theta^{2} e^{\theta d_{i}} (1-e^{\theta d_{i}})}{(1 + e^{\theta d_{i}})^3} < 0 . \label{Appendixkey1} 
		\end{align}  
%		Hence, for 
%		\begin{align}
%		Q(\xibold) =& \dfrac{1}{\Fxi}, \nonumber \\
%		\dfrac{\partial Q(\xibold)}{\partial d_{l}} - \dfrac{\partial Q(\xibold)}{\partial d_{k}} & = \dfrac{1}{(\Fxi)^{2}} \Big[ \dfrac{\partial \omega(d_{k})}{\partial d_{k}} - \dfrac{\partial \omega(d_{l})}{\partial d_{l}}  \Big] \text{ for } k,l = 1,\ldots,n-1.\label{Appendixkey2} 
%		\end{align}  
		Hence, for $ Q(\xibold) =\dfrac{1}{\Fxi}$ we have 
				\begin{align}\dfrac{\partial Q(\xibold)}{\partial d_{l}} - \dfrac{\partial Q(\xibold)}{\partial d_{k}} & = \dfrac{1}{(\Fxi)^{2}} \Big[ \dfrac{\partial \omega(d_{k})}{\partial d_{k}} - \dfrac{\partial \omega(d_{l})}{\partial d_{l}}  \Big] \text{ for } k,l = 1,\ldots,n-1.\label{Appendixkey2} 
		\end{align}  
		Note that $Q(\cdot)$ is permutation invariant of $d_{i}$'s . Also, $ \dfrac{\partial Q(\xibold)}{\partial d_{l}} > \dfrac{\partial Q(\xibold)}{\partial d_{k}} \text{ for } d_{l} > d_{k} $, where $k,l = 1,\ldots,n-1$ (using equations \eqref{Appendixkey1} and \eqref{Appendixkey2}). So, we can say that $Q(\cdot)$ is a Schur-convex function (from Theorem A.4 in \cite{MarshalOlkinBook}) and hence it is minimized for an equispaced design, that is $d_{i} = \dfrac{1}{n-1}$ for all $i$. \\

		\section{Appendix} \label{AppendixC}
		In this part, we look at some matrix and vector decompositions used for proving results involving the $SMSPE$ for simple and ordinary cokriging models. 
		
		Consider $x_{0} \in [x_{i}, x_{i+1}]$ for some $ i = 1,\ldots,n-1$, recall that $ a = x_{0} - x_{i}$ and let an $n \times n $ diagonal matrix, $\MatrixE = 	diag 
		\begin{pmatrix} 
		e^{- \theta \sum_{l=1}^{i-1} d_{l}}, &
		e^{- \theta \sum_{l=2}^{i-1} d_{l} },&
		\ldots
		1,&
		1,&
		e^{- \theta d_{i+1}},&
		\ldots
		e^{- \theta \sum_{l=i+1}^{n-1} d_{l}}  
		\end{pmatrix} 
		$, such that $(\MatrixE)_{ii} = 1$ and $(\MatrixE)_{i+1\;i+1} = 1$. Also, consider two vectors of length $n$, $\VectEOne =  
		\begin{pmatrix} 
			1 & 
			1&
			\ldots		1 &
			0 &
			0&
			\ldots	0  
		\end{pmatrix} $ and  $\VectETwo = \begin{pmatrix} 
		0&
		0&
		\ldots		0 &
		1 &
		1&
		\ldots	1  
		\end{pmatrix}
		$, such that $(\VectEOne)_{i} = 1 $ and  $ (\VectEOne)_{i+1} = 0 $, and $(\VectETwo)_{i} = 0 $ and $ (\VectETwo)_{i+1} = 1 $. Then, we may write $\sigmaPNot$ as, 
		\begin{align}
	\sigmaPNot  &= \MatrixE \;\;  \Big[  e^{-\theta a } \VectEOne + 	e^{-\theta (d_{i} - a) } \VectETwo\Big].
	 \label{DecompositionSigmaNot}
	\end{align}	
		
%Keep this in thesis 
%Keep this in thesis 
%Keep this in thesis 
%Keep this in thesis 

%That is,
%	\begin{align*}
%	\sigmaPNot  &= 
%	diag \; 
%	\begin{pmatrix} 
%	e^{- \theta \sum_{l=1}^{i-1} d_{l}}\\
%	e^{- \theta \sum_{l=2}^{i-1} d_{l} }\\
%	.\\
%	.\\
%	.\\
%	1 \\
%	1 \\
%	e^{- \theta d_{i+1}}\\
%	.\\
%	.\\
%	.\\
%	e^{- \theta \sum_{l=i+1}^{n-1} d_{l}}  
%	\end{pmatrix} 
%	%	\(
%	\begin{pmatrix}
%			e^{-\theta a }
%			\begin{pmatrix} 
%			1\\
%			1\\
%			.\\
%			.\\
%			.\\
%			1 \\
%			0 \\
%			0\\
%			.\\
%			.\\
%			.\\
%			0  
%			\end{pmatrix} 
%			+ 
%			e^{-\theta (d_{i} - a) }
%			\begin{pmatrix} 
%			0\\
%			0\\
%			.\\
%			.\\
%			.\\
%			0 \\
%			1 \\
%			1\\
%			.\\
%			.\\
%			.\\
%			1  
%			\end{pmatrix}
%		\end{pmatrix} .
%	\end{align*}	
%	
%	
Using the $n \times 1$ vectors $\VectUOne, \VectUTwo, \VectVOne$, and $ \VectVTwo$ defined as:
		\begin{align}
		\VectUOne^{T}	&= 
		\begin{pmatrix} e^{- \theta \sum_{l=1}^{i-1} d_{l}},& 
		e^{- \theta \sum_{l=2}^{i-1} d_{l} } ,& \ldots &  & \;\;,1^{i^{th} pos},& 0,& 0,\ldots& \;\;\;\;\;\;\;\;\;\;\;\;\;\;\;\;\;\;\;\;\;\;\;\;\;\;\;\;\;\; &&\ldots\ldots,0
		\end{pmatrix}, \nonumber \\
		\VectUTwo^{T} &= 
		\begin{pmatrix} 0, & 0,\ldots &	 \;\;\;\;\;\;\;\;\;\;\;\;\;\;\;\;\;\;\;\;\;\;\;\; & 	&\ldots ,0\;\;\;\;&,1^{(i+1)^{th} pos},& e^{- \theta d_{i+1}},  &
		& 	 & \;\ldots\ldots,e^{- \theta \sum_{l=i+1}^{n-1} d_{l}}  
		\end{pmatrix} ,\nonumber \\
		\VectVOne^{T}	&= 
		\begin{pmatrix} 0, & 
		0,\ldots &  & \;\;\; &,0^{(i-1)^{th}pos}, & \dfrac{1}{1-e^{-2\theta d_{i}}},& \dfrac{-e^{-\theta d_{i}}}{1-e^{-2\theta d_{i}}},& 0,& \;\;\;\;\;\;\;\;\;\;\;\;\;\;\;\;\; &\;\;\ldots\ldots\ldots,0
		\end{pmatrix}, \nonumber \\
		\VectVTwo^{T}	&= 
		\begin{pmatrix} 0, & 
		0,\ldots & & \;\;\; &,0^{(i-1)^{th}pos}, & \dfrac{-e^{-\theta d_{i}}}{1-e^{-2\theta d_{i}}},& \dfrac{1}{1-e^{-2\theta d_{i}}},& 0,& \;\;\;\;\;\;\;\;\;\;\;\;\;\;\;\;\;\;\; &\ldots\ldots\ldots,0
		\end{pmatrix} \nonumber 
		\end{align}
		and $\sigmaPNot $ from \eqref{DecompositionSigmaNot} and $\MatrixP^{-1}$ from \eqref{CovMatInverse}  we obtain:
	%Using simple matrix calculations, the value of $\MatrixP^{-1}$ as in equation \eqref{CovMatInverse} and equation \eqref{DecompositionSigmaNot} it could be easily verified that:
		\begin{align}
		\sigmaPNot 
		&= 
		e^{-\theta a } \VectUOne + e^{-\theta (d_{i}-a) } \VectUTwo, \label{decomposition1}\\
		\MatrixP^{-1} \sigmaPNot 
		&= 
		e^{-\theta a } 
		\VectVOne 
		+ e^{-\theta (d_{i}-a) }
		\VectVTwo, \label{decomposition2}\\
		\sigmaPNot^{T} \MatrixP^{-1} \sigmaPNot &= \dfrac{ e^{-2\theta a } -2  e^{-2\theta d_{i}   } + e^{-2\theta (d_{i} - a ) }   }{ 1- e^{-2\theta d_{i} }   },  \text{ and }  \label{decomposition3} \\ %\label{SMSPE_Decomposition_Eq_I}
		\VectOneN^{T} \MatrixP^{-1} \sigmaPNot  &= \dfrac{e^{- \theta a } + e^{- \theta(d_{i}  - a )  } }{ 1 + e^{- \theta d_{i}  }}. \label{decomposition4} %\label{SMSPE_Decomposition_Eq_II} 
		\end{align}
		
	\section{Appendix} \label{AppendixD}
We show here that if $x_0 \in [x_i, x_{i+1}]$ for some $i = 1,...,n-1$ then the $MSPE_{sk}(x_{0})$ is maximized at $x_{0} = x_{i} + \dfrac{d_{i}}{2}$. %which is used in Theorem~\ref{Theorem_SK_SMSPE}.
From Lemma~\ref{LemmaSimpleAndOrdinary}, we have 
	\begin{align}
	MSPE_{sk}(x_{0}) 
	& = \sigma_{11} \dfrac{\left( 1   - e^{-2 \theta a}\right) \left( 1   - e^{-2 \theta (d_{i} - a)}\right)}{\left( 1   - e^{-2 \theta d_{i}}\right)} \label{Thm1_eq1}.	
	\end{align}  
			Since, $x_0 \in [x_i, x_{i+1}]$ and $ a = x_{0} - x_{i} $, therefore $a \in [0,d_{i}]$ for $i=1,\ldots, n-1$. Now, consider the function  		
			\begin{flalign}
			W_{i}:& \;\; [0,d_{i}] \; \to \mathbb{R} \;\;\;\;\;\;\;\;\;\;\;\;\;  \nonumber \\%\mbox{ such that }	 \\
			a & \mapsto	  \dfrac{\left( 1   - e^{-2 \theta a}\right) \left( 1   - e^{-2 \theta (d_{i} - a)}\right)}{\left( 1   - e^{-2 \theta d_{i}}\right)}.  \nonumber	 
			\end{flalign}
			Differentiating $W_{i}(\cdot)$ with respect to $a$ we get, 
			\begin{align}
			& \dfrac{d W_{i}(a)}{ d a }  = \dfrac{2\theta\left( e^{-2 \theta a}   - e^{-2 \theta (d_{i} - a)}\right)}{\left( 1   - e^{-2 \theta d_{i}}\right)},   \nonumber	   
			\end{align}
			where, 
			\begin{align}
			& \dfrac{d W_{i}(a)}{ d a }\Bigg|_{a = d_{i}/2} = 0, \label{Thm1_eq2} 
			\end{align}  
			and 
			\begin{align}
			\dfrac{d^{2} W_{i}(a)}{ d a^{2} }   &=  \dfrac{ -4 \theta^2	 \left( e^{-2 \theta a}   + e^{-2 \theta (d_{i} - a)}\right)}{\left( 1   - e^{-2 \theta d_{i}}\right)} < \; 0.   \label{Thm1_eq3}
			\end{align} 
	From equations \eqref{Thm1_eq2} and \eqref{Thm1_eq3}, for $x_0 \in [x_i, x_{i+1}]$, $W_{i}(\cdot)$ is maximized at $d_{i}/2$ or equivalently $MSPE(x_{0})$ over $[x_i, x_{i+1}]$ is maximized at $x_{0} = x_{i} + \dfrac{d_{i}}{2}$. %, which is the mid-point of the interval $[x_i, x_{i+1}]$. \\
	 Hence, 
	\begin{align}
	\sup_{x_0 \in [x_{i}, x_{i+1}]} MSPE(x_{0}) = & W_{i}(d_{i}/2) \nonumber \\
	 = & \sigma_{11} \frac{1 - e^{-\theta d_i}}{1 + e^{-\theta d_i}}. \label{Thm1_eq4}
	\end{align}
	
	\section{Appendix} \label{AppendixE}
	We prove that $IMSPE_{sk}$ is a Schur-convex function. First note, $IMSPE$ is a symmetric function, that is, it is  permutation invariant in the $d_{i}$'s. Next we find $\dfrac{\partial IMSPE}{\partial d_{i}} $ and show that it is an increasing function in the $d_{i}$'s for $i= 1,\ldots,n$; 
	\begin{align}
	\dfrac{\partial \phi(d)}{\partial d} &= \dfrac{e^{2\theta d}-1-2\theta d e^{2\theta d}}{(e^{2\theta d}-1)^2}  \text{ is an increasing function in } d \in (0,1) \label{SchurConvex_SimpleKriging_III} \\
	%\text{and,  }
	%& \dfrac{\partial^{2} \phi(d)}{\partial d^{2}} = \dfrac{4e^{2\theta d}}{{(e^{2\theta d}-1)}^3} (1+\theta d + e^{2\theta d} (\theta d-1))  \geq 0 \label{PhiConvexEq2}
	\text{since, }
	\dfrac{\partial^{2} \phi(d)}{\partial d^{2}} &= \dfrac{4 \theta e^{2\theta d}}{{(e^{2\theta d}-1)}^3} (1+\theta d + e^{2\theta d} (\theta d-1)) \nonumber \\
	&= \dfrac{4 \theta e^{2\theta d}}{{(e^{2\theta d}-1)}^3} \;\; p(d,\theta)  \geq 0, \text{ for } d \in (0,1), \nonumber 
	\end{align}
	where $p(d,\theta)=(1+\theta d + e^{2\theta d} (\theta d-1))\geq 0$ and $\frac{\partial p(d)}{\partial d}|_{d=0}=\frac{\partial^2 p(d)}{\partial d^2}|_{d=0}=0 $ and $\frac{\partial^2 p(d)}{\partial d^2}>0$ for $d \in (0,1]$.\\% So, $\dfrac{\partial^{2} \phi(d)}{\partial d^{2}} \geq 0 $ for $d \in (0,1).$\\
	%			And $\dfrac{\partial^{2} \phi(d)}{\partial d^{2}} \geq 0 $ because $\frac{\partial p(d)}{\partial d}|_{d=0}=\frac{\partial^2 p(d)}{\partial d^2}|_{d=0}=0 $ and $\frac{\partial^2 p(d)}{\partial d^2}>0$ for $d \in [0,1]$.   \\
	As, $ \dfrac{\partial IMSPE}{\partial d_{i}} = 2  \sigma_{11} \dfrac{\partial \phi(d_{i})}{\partial d_{i}} $ for $i=1,\ldots, n-1$, using \eqref{SchurConvex_SimpleKriging_III} we can say:
	\begin{equation}
	\dfrac{\partial IMSPE}{\partial d_{k}} 
	\leq \dfrac{\partial IMSPE}{\partial d_{l}} \;\;\; \text{ for any }  d_{k} \leq d_{l}.  \label{SchurConvex_SimpleKriging_I} \\	
	\end{equation}
	Thus, using Theorem A.4 from \cite{MarshalOlkinBook}, we can say that $IMSPE$ is Schur-convex. %Hence IMSPE is minimized for an equispaced design, i.e., $d_i=\frac{1}{n-1}$ for all $i=1,\ldots, n-1$.
\section{Appendix} \label{AppendixF}
We show that for $x_0 \in [x_{i},x_{i+1}]$ for some $i = 1,\ldots, n-1$, $\smash{\displaystyle \sup_{x_0 \in [x_{i},x_{i+1}]} } \Big( 1 -  \VectOneN^{T} \MatrixP^{-1} \sigmaPNot \Big)^{2} $ is attained at $x_{0} = x_{i} + \dfrac{d_{i}}{2}$. From \eqref{decomposition4} in \ref{AppendixC} we have,
%Consider the case where $x_0 \in [x_{i},x_{i+1}]$ where, $i = 1,\ldots, n-1$. First we show that $\sup_{x_0 \in [x_{i},x_{i+1}]} \Big( 1 -  \VectOneN^{T} \MatrixP^{-1} \sigmaPNot \Big)^{2} $ is attained at $x_{0} = x_{i} + \dfrac{d_{i}}{2}$. Take $a= x_{0}-x_{i}$, then $a \in [0,d_{i}]$ then,
\begin{align}
\VectOneN^{T} \MatrixP^{-1} \sigmaPNot  &= \dfrac{e^{- \theta a } + e^{- \theta(d_{i}  - a )  } }{ 1 + e^{- \theta d_{i}  }}. \nonumber 
\end{align}
As $a \in [0,d_{i}]$, defining the function $U_{i}(\cdot)$ such that, 	
\begin{align}
U_{i}:&  [0,d_{i}] \to \mathbb{R} \nonumber \\
a & \mapsto \Bigg(1 -  \dfrac{e^{- \theta a } + e^{- \theta( d_{i} - a )  } }{ 1 + e^{- \theta d_{i}  }} \Bigg)^2 \nonumber 		
\end{align}
we obtain 	
%	\begin{flalign}
%	& \dfrac{d U_{i}(a)}{ d a }  = -2 \theta \underbrace{  \Bigg(1 -  \dfrac{e^{- \theta a } + e^{- \theta(d_{i} - a )  } }{ 1 + e^{- \theta d_{i}  }} \Bigg)}_{Term I} \underbrace{ 		\Bigg( \dfrac{ -e^{- \theta a}   + e^{-  \theta (d_{i} - a)}}{ 1   +  e^{- \theta d_{i}}} \Bigg) }_{Term II} & \label{Thm_ock_eq2} \\
%	&	\dfrac{d U_{i}(a)}{ d a }\Bigg|_{a = d_{i}/2} = 0 & \label{Thm_ock_eq3} \\
%	&	\dfrac{d^{2} U_{i}(a)}{ d a^{2} }   =  -4 \theta^{2} \Big( \dfrac{1-e^{-\theta	d_{i}/2}}{1+e^{-\theta d_{i}}} \Big)^2 e^{-\theta d_{i}}< 0 & \label{Thm_ock_eq4}
%	\end{flalign} 
\begin{align}
& \dfrac{d U_{i}(a)}{ d a }  = -2 \theta \underbrace{  \Bigg(1 -  \dfrac{e^{- \theta a } + e^{- \theta(d_{i} - a )  } }{ 1 + e^{- \theta d_{i}  }} \Bigg)}_{Term I} \underbrace{ 		\Bigg( \dfrac{ -e^{- \theta a}   + e^{-  \theta (d_{i} - a)}}{ 1   +  e^{- \theta d_{i}}} \Bigg) }_{Term II}  \label{Thm_ock_eq2} 
\end{align}
where
\begin{align}
&	\dfrac{d U_{i}(a)}{ d a }\Bigg|_{a = d_{i}/2} = 0  \label{Thm_ock_eq3} 
\end{align}
and 
\begin{align}
&	\dfrac{d^{2} U_{i}(a)}{ d a^{2} }   =  -4 \theta^{2} \Big( \dfrac{1-e^{-\theta	d_{i}/2}}{1+e^{-\theta d_{i}}} \Big)^2 e^{-\theta d_{i}}< 0 . \label{Thm_ock_eq4}
\end{align} 
From \eqref{Thm_ock_eq3} and \eqref{Thm_ock_eq4} we see $U_{i}(\cdot)$ attains a local maxima at $a = \dfrac{d_{i}}{2}$ and $U_{i}( \dfrac{d_{i}}{2}) = \Big( 1- \dfrac{2e^{- \theta	d_{i}/2}}{1+ e^{- \theta	d_{i}}}\Big)^2 > 0 $. To find the point of maxima $a = d_{i}/2$ we set $Term\;II$ in    \eqref{Thm_ock_eq2} equal to zero. Any other point $a_{1}$ at which $U^{\prime}(a_{1})=0$ is obtained by setting $Term\;I$ equal to zero; however, those points could not be the maxima as $U_{i}(a_{1})$ is zero. \\
Hence, we have shown that $\smash{\displaystyle \sup_{a \in [0,d_{i}]} }
 U_{i}(a) = \smash{\displaystyle \sup_{x_0 \in [x_{i},x_{i+1}]} } \Big( 1 -  \VectOneN^{T} \MatrixP^{-1} \sigmaPNot \Big)^{2}$ is attained at $a = \dfrac{d_{i}}{2}$ or $x_{0} = x_{i} + \dfrac{d_{i}}{2}$ for some $i=1,\ldots,n-1$, which is the mid-point of the interval $[x_{i},x_{i+1}]$. \\

Hence, we obtain 
\begin{align}
\sup_{x_0 \in [x_{i},x_{i+1}]} \Big( 1 -  \VectOneN^{T} \MatrixP^{-1} \sigmaPNot \Big)^{2} 
&= U_{i}( \dfrac{d_{i}}{2}) = \Big( 1- \dfrac{2e^{- \theta	d_{i}/2}}{1+ e^{- \theta	d_{i}}}\Big)^2 \label{EqA}
\end{align}
 
	\end{appendix}
	
	\section*{Acknowledgements}
	The authors would like to thank Prof. Subhankar Karmakar (Centre for Environmental Science and Engineering, IIT Bombay) for the data and the image of the river.
	
	\section*{Funding }
This project is funded by IITB-Monash Research Academy, India. \\
The work of S. Mukhopadhyay was supported by the Science and Research Engineering Board (Department of Science and Technology, Government of India) [File Number: EMR/2016/005142].

	%\begin{supplement} \label{SupplementTwo}
	%	\textbf{Supplement II}.
	%	This file has tables and figures corresponding to Section~\ref{Illustration_cokriging_known_parameter} and \ref{OptimalDesign_OrdinaryCokrig}. 
	%\end{supplement}

	%%  MR numbers will be added by VTeX.                      %%
	%%                                                         %%
	%%  Use \cite{...} to cite references in text.             %%
	%%                                                         %%
	%%%%%%%%%%%%%%%%%%%%%%%%%%%%%%%%%%%%%%%%%%%%%%%%%%%%%%%%%%%%%
	
	%if your bibliography is in bibtex format, uncomment commands:
	%\bibliographystyle{elsarticle-num} % Style BST file
	%\bibliographystyle{imsart-number} % Style BST file (imsart-number.bst or imsart-nameyear.bst)
%	\begin{thebibliography}{00}
%	\end{thebibliography}

 \bibliographystyle{elsarticle-num} 
\bibliography{citebib__4}

	       % Bibliography file (usually '*.bib')
	
	%% or include bibliography directly:
	%\bibliography{citebib__2}
	%\begin{thebibliography}{4}
	%%%
	%\bibitem{r1}
	%\textsc{Billingsley, P.} (1999). \textit{Convergence of
	%Probability Measures}, 2nd ed.
	%Wiley, New York.
	%
	%\bibitem{r2}
	%\textsc{Bourbaki, N.}  (1966). \textit{General Topology}  \textbf{1}.
	%Addison--Wesley, Reading, MA.
	%
	%\bibitem{r3}
	%\textsc{Ethier, S. N.} and \textsc{Kurtz, T. G.} (1985).
	%\textit{Markov Processes: Characterization and Convergence}.
	%Wiley, New York.
	%
	%\bibitem{r4}
	%\textsc{Prokhorov, Yu.} (1956).
	%Convergence of random processes and limit theorems in probability
	%theory. \textit{Theory  Probab.  Appl.}
	%\textbf{1} 157--214.
	%\end{thebibliography}
		\section*{Correspondence}
		Prof. Siuli Mukhopadhyay, \\
		Department of Mathematics, Indian Institute of Technology Bombay,\\
		 Mumbai, Maharashtra, 400076, India\\
	email - siuli@math.iitb.ac.in

\end{document}